%% file: main.tex
\documentclass[11pt]{article}

\usepackage{geometry}
\usepackage[utf8]{inputenc}
\usepackage[english]{babel}
\usepackage{lmodern}
\usepackage{microtype}

\usepackage{graphicx}
\usepackage{authblk}
\usepackage{hyperref}
\usepackage{xcolor}
\usepackage{standalone}

\usepackage{amsmath,amsfonts,amssymb}
\usepackage{amsthm}
\usepackage[capitalise]{cleveref}
\usepackage[nocompress]{cite}

\title{Compositionality of planar perfect matchings\\ \normalsize A universal and complete fragment of ZW-calculus}

\author{Titouan Carette}
\affil{\href{mailto:titouan.carette@lu.lv}{titouan.carette@lu.lv}\\
  Centre for Quantum Computer Science, Faculty of Computing, University of Latvia, Raina 19, Riga, Latvia, LV-1586}

\author{Etienne Moutot}
\affil{\href{mailto:etienne.moutot@math.cnrs.fr}{etienne.moutot@math.cnrs.fr}\\
CNRS, I2M, Aix-Marseille Université, Marseille, France}

\author{Thomas Perez}
\affil{\href{malto:thomas.perez@ens-lyon.fr}{thomas.perez@ens-lyon.fr}\\
Université de Lyon, ENS de Lyon, 69007 Lyon, France}

\author{Renaud Vilmart}
\affil{\href{mailto:vilmart@lsv.fr}{vilmart@lsv.fr}
Université Paris-Saclay, ENS Paris-Saclay, Inria, CNRS, LMF, 91190, Gif-sur-Yvette, France}

\date{}

\usepackage[utf8]{inputenc}
\usepackage{amssymb}
\usepackage{amsfonts}
\usepackage{mathtools}
\usepackage{calc}
\usepackage{hyperref}
\usepackage{cleveref}
\usepackage{graphicx}
\usepackage{multicol}
\usepackage{amsmath}
\usepackage{stmaryrd}
\usepackage{multicol}
\usepackage{enumerate}
\usepackage{todonotes}
\usepackage{physics}

\usepackage{tikz}
\usepackage{pgfplots}
\usetikzlibrary{decorations.markings,decorations.pathmorphing, decorations.text,positioning}
\usetikzlibrary{backgrounds,folding}
\usetikzlibrary{shapes.geometric, shapes.misc}
\usetikzlibrary{arrows.meta}
\pgfdeclarelayer{edgelayer}
\pgfdeclarelayer{nodelayer}
\pgfsetlayers{background,edgelayer,nodelayer,main}
\tikzstyle{every picture}=[baseline=-0.25em]
\tikzset{every path/.style={draw=black!80, line width=0.6pt}}
\tikzstyle{none}=[inner sep=0mm]
\tikzstyle{black dot}=[inner sep=0.5mm,minimum width=0pt,minimum height=0pt,fill=black,draw=black,shape=circle]
\tikzstyle{dot}=[black dot]
\tikzstyle{white dot}=[inner sep=0.5mm,minimum width=0pt,minimum height=0pt,fill=white,draw=black,shape=circle]
\tikzstyle{white phase dot}=[minimum size=4mm, font={\footnotesize\boldmath}, shape=rectangle, rounded corners=1.5mm, inner sep=0.8mm, outer sep=-2mm, scale=0.8, draw=black, fill=white]
\tikzstyle{H box}=[fill=yellow, draw=black, shape=rectangle, inner sep=0.6mm, minimum height=1.5mm, minimum width=1.5mm]
\tikzstyle{box}=[rectangle,fill=white,draw=black, font=\scriptsize, inner sep=2pt]
\tikzstyle{box-no-outline}=[rectangle, inner sep=2pt]
\tikzstyle{fswap}=[inner sep=0.7mm,minimum width=0pt,minimum height=0pt,draw=purple,shape=circle]
\tikzset{
	tickedge/.style={
		decoration={ markings,
			mark=at position .5 with {\draw (0,2pt) -- (0,-2pt);}
		},
		postaction={decorate}
	},
}
\tikzstyle{compact dash}=[dash pattern={on 2pt off 1pt}]
\usetikzlibrary{arrows}
\tikzstyle{dashed arrow}=[->, draw=black, dashed]
\tikzstyle{hook arrow}=[right hook->]
\tikzstyle{arrow}=[->]

\pgfdeclarelayer{edgelayer}
\pgfdeclarelayer{nodelayer}
\pgfsetlayers{background,edgelayer,nodelayer,main}
\tikzstyle{every loop}=[]

\usepackage{etoolbox}

\newtoggle{extern}
\togglefalse{extern}

\newcommand{\tikzfig}[2][]{
	{\setlength{\fboxsep}{0pt}\colorbox{gray!15}{~#1\strut\input{./figures/#2.tikz}#1}}
}

\def\fig{}

\iftoggle{extern}{
	\usetikzlibrary{external}
	\tikzexternalize[prefix=tikzfigs/]
	
	\renewcommand{\tikzfig}[1]{
		\tikzsetnextfilename{#1}
		
	{\setlength{\fboxsep}{0pt}\colorbox{gray!15}{~\strut\input{./figures/#1.tikz}}}
}
	
}{
	
}

\newcommand{\eq}[2][]{
	#1
	\underset{\substack{#2}}{=}
	#1
}

\newcommand{\interp}[1]{\left\llbracket #1 \right\rrbracket}

\newcommand{\cat}[1]{\mathbf{#1}}
\newcommand{\pW}{\operatorname{p}\!\!\cat{W}}

\newcommand{\bvdots}{ \tikz[baseline, every node/.style={inner sep=0}]{ \node at (0,0){.}; \node at (0,-3pt){.}; \node at (0,3pt){.}; } }

\allowdisplaybreaks
\makeatletter
\renewcommand{\todo}[2][]{\tikzexternaldisable\@todo[#1]{#2}\tikzexternalenable}

\newcommand{\N}{\ensuremath{\mathbb{N}}}

\newcommand{\C}{\ensuremath{\mathbb{C}}}

\theoremstyle{plain}
\newtheorem{theorem}{Theorem}[section]
\newtheorem{lemma}[theorem]{Lemma}

\newtheorem{proposition}[theorem]{Proposition}
\crefname{propproposition}{Proposition}{Propositions}

\theoremstyle{definition}
\newtheorem{definition}{Definition}[section]
\theoremstyle{remark}

\newtheorem{example}[theorem]{Example}

\begin{document}
	
	\maketitle

\begin{abstract}
	
	We exhibit a strong connection between the matchgate formalism introduced by Valiant and the ZW-calculus of Coecke and Kissinger. This connection provides a natural compositional framework for matchgate theory as well as a direct combinatorial interpretation of the diagrams of ZW-calculus through the perfect matchings of their underlying graphs.
	
	We identify a precise fragment of ZW-calculus, the planar W-calculus, that we prove to be complete and universal for matchgates, that are linear maps satisfying the matchgate identities. Computing scalars of the planar W-calculus corresponds to counting perfect matchings of planar graphs, and so can be carried in polynomial time using the FKT algorithm, making the planar W-calculus an efficiently simulable fragment of the ZW-calculus, in a similar way that the Clifford fragment is for ZX-calculus. This work opens new directions for the investigation of the combinatorial properties of ZW-calculus as well as the study of perfect matching counting through compositional diagrammatical technics.
	
\end{abstract}


\section{Introduction}

A quantum computation mapping $n$ qubits to $m$ qubits corresponds to an isometric linear map $\mathbb{C}^{2^n} \to \mathbb{C}^{2^m}$. Due to the exponential size of their matrix representation, those linear maps are traditionally depicted as quantum circuits, an assemblage of elementary quantum gates similar to the more common boolean circuits. Given a quantum circuit $n\to m$, evaluating a coefficient of the corresponding $2^m \times 2^n $ matrix (i.e. evaluating the circuit with a given input) typically requires an exponential time. However, there are some specific classes of quantum circuits -- or fragments --, that can be classically simulated in polynomial time. Examples are the Clifford fragment (as asserted by the Gottesman-Knill theorem) as well as the fragment that will particularly interest us in this paper, the nearest-neighbour matchgates \cite{DBLP:journals/siamcomp/Valiant02}. Investigating those tractable fragments allows a better understanding of the computational advantage of quantum computing. The reference for all elementary results on quantum circuits is \cite{nielsen2002quantum}.  

Taking the diagrammatical circuit representation seriously led to developing graphical languages for quantum computing \cite{DBLP:conf/icalp/CoeckeD08}. Those languages are equational theories described by elementary gates and local identities between diagrams. Such languages come with an interpretation into linear maps. A language is said universal for a class of linear maps if any linear map in the class is the interpretation of a diagram in the language.
A language is said complete if two diagrams with the same interpretation are equivalent up to the equational theory, which means that they can be rewritten from one to the other using the local rewriting rules of the equational theory. In general, completeness is the most challenging property to prove.

The first quantum graphical language to appear was the ZX-calculus in 2008 \cite{DBLP:conf/icalp/CoeckeD08}. It was rapidly known to be universal for all linear maps. However, providing a complete set of rewriting rules took another ten years (see \cite{van2020zx} for an history of completeness) and first required a translation through another language, the ZW-calculus \cite{ng2017universal}.

The ZW-calculus was introduced in \cite{DBLP:conf/icalp/CoeckeK10} as a graphical representation of the two kinds of tripartite entanglement for qubits, namely the GHZ-states and W-states.
It then appeared that this calculus had very nice algebraic properties allowing the internal encoding of arithmetical operations.
Those properties allowed the ZW-calculus to be the first proven universal and complete language for linear maps \cite{DBLP:conf/lics/Hadzihasanovic15}. Despite this historical importance, the ZW-calculus gathered less attention than other languages, seen as more connected to quantum computing. Still, we must mention interesting connections with fermionic quantum computing \cite{DBLP:journals/lmcs/FeliceHN19}, and recent works importing some ZW-calculus primitives into ZX-calculus to exploit their algebraic properties \cite{shaikh2022sum,wang2022differentiating}. In this paper, we show that ZW-calculus has very strong connections with a specific family of quantum circuits: the matchgates.

Matchgates were introduced in 2002 by Valiant \cite{DBLP:journals/siamcomp/Valiant02}. They are linear maps defined by counting the perfect matching of a graph from which we remove some vertices depending on the inputs. This underlying combinatorial structure allows to classically simulate the corresponding quantum circuits by using the Polynomial FKT algorithm for perfect matchings counting \cite{doi:10.1080/14786436108243366,kasteleyn1961statistics}. The theory of matchgates was then developed further to the concept of holographic algorithms \cite{valiant2008holographic}. We can notice that if some connections between graphical languages and holographic algorithms have been investigated \cite{backens2017new}, we are not aware of any diagrammatical approach to the original concept of matchgate before the present work, except a mention in \cite{DBLP:journals/lmcs/FeliceHN19}.

The main contribution of this paper is the introduction of a fragment of the ZW-calculus, that we call planar W-calculus. We show that this language is universal and complete for the planar matchgate fragment of quantum computation.
The completeness proof relies on designing a normal form and a rewriting strategy to reach it. We also define a pro of matchgate computations by showing the compositionality of the matchgate identities introduced in \cite{DBLP:journals/mst/CaiCL09}. The combinatorial characterisation of matchgate computations then directly follows from the correspondence with the graphical language. Hence one can see this paper as a reformulation of matchgate theory in a compositional framework.

The paper is structured as follows. \Cref{sec:2} introduces our graphical primitives, their interpretation as linear maps and their combinatorial properties: the interpretation of a diagram can be deduced by counting the number of perfect matching of the underlying weighted graph. We present the generators and elementary rewrite rules of the language as well as an essential syntactic sugar: the fermionic swap that emulates the swap gate, which is not part of our language.
\Cref{sec:3} introduces the normal form and proves the completeness of the language.
In \Cref{sec:4}, we properly define a pro of matchgates characterised as the linear maps satisfying the matchgate identities. We show that our language is universal for matchgates, \textbf{i.e.}, that the interpretation of a diagram is always a matchgate and that all matchgates correspond to a diagram.
Finally, in \Cref{sec:5}, we sketch future directions of research suggested by the connection we identified between ZW-calculus and perfect matching counting.

\section{Perfect Matchings and Planar W-Calculus}\label{sec:2}
We define our fragment of the ZW-calculus, the \emph{planar W-calculus}, by defining its diagrams. Any diagram with $n$ inputs and $m$ outputs $D:n \rightarrow m$ is interpreted as a linear map $\interp{D}: \C^{2^n}\rightarrow \C^{2^m}$ inductively as follows:
\begin{center}
  $\interp{\def\fig{def}{\setlength{\fboxsep}{0pt}\colorbox{gray!15}{~\strut\input{./figures/\fig/\fig_tensor.tikz}}}} := \interp{D_1} \otimes\interp{D_2}$ \quad
  $\interp{\def\fig{def}{\setlength{\fboxsep}{0pt}\colorbox{gray!15}{~\strut\input{./figures/\fig/\fig_compose.tikz}}}} := \interp{D_2} \circ\interp{D_1}$ \\
  $\interp{\def\fig{rewrite-zero}{\setlength{\fboxsep}{0pt}\colorbox{gray!15}{~\strut\input{./figures/\fig/\fig_01.tikz}}}} := (1)$ \quad
  $\interp{\def\fig{def}{\setlength{\fboxsep}{0pt}\colorbox{gray!15}{~\strut\input{./figures/\fig/\fig_identity.tikz}}}} :=
  \begin{pmatrix}
    1 & 0 \\
    0 & 1
  \end{pmatrix}
  $ \quad 
  $\interp{\def\fig{def}{\setlength{\fboxsep}{0pt}\colorbox{gray!15}{~\strut\input{./figures/\fig/\fig_cup.tikz}}}} :=
  \begin{pmatrix}
    1 & 0 & 0 & 1
  \end{pmatrix}
  $ \quad
  $\interp{\def\fig{def}{\setlength{\fboxsep}{0pt}\colorbox{gray!15}{~\strut\input{./figures/\fig/\fig_cap.tikz}}}} :=
  \begin{pmatrix}
    1 \\ 0 \\ 0 \\ 1
  \end{pmatrix}
  $
\end{center}
In particular, note that we do \emph{not} use the usual swap diagram $\def\fig{def}{\setlength{\fboxsep}{0pt}\colorbox{gray!15}{~\strut\input{./figures/\fig/\fig_swap.tikz}}}$, hence the name \emph{planar}. We do have, however, the so-called cup $\def\fig{def}{\setlength{\fboxsep}{0pt}\colorbox{gray!15}{~\strut\input{./figures/\fig/\fig_cup.tikz}}}$ and cap $\def\fig{def}{\setlength{\fboxsep}{0pt}\colorbox{gray!15}{~\strut\input{./figures/\fig/\fig_cap.tikz}}}$ satisfying the ``snake equations'':
\def\fig{snake}
\begin{align*}
{\setlength{\fboxsep}{0pt}\colorbox{gray!15}{~\strut\input{./figures/\fig/\fig_00.tikz}}}
\eq{}{\setlength{\fboxsep}{0pt}\colorbox{gray!15}{~\strut\input{./figures/\fig/\fig_01.tikz}}}
\eq{}{\setlength{\fboxsep}{0pt}\colorbox{gray!15}{~\strut\input{./figures/\fig/\fig_02.tikz}}}
\end{align*}
In the following, with $D:n\to n$, we may use the following notation: $D^{\otimes\vec b}$ when $\vec b$ is a bitstring, to represent $D^{b_1}\otimes...\otimes D^{b_n}$ with $D^0 = id_n$ and $D^1 = D$. We call a diagram $D$ a scalar if it has no input and no output, i.e.~$D:0\to0$. In the category-theoretic terminology, such a collection of diagrams defines a pro, a strict monoidal category whose monoid of objects is generated by a unique element, and not a prop, which requires the category to be symmetric, \textit{i.e.} to have swap diagrams. Furthermore, the presence of the cups and caps make the category a compact-closed pro. 
We define $\cat{Qubit}$ to be the prop whose $n\to m$ morphisms are linear maps $\C^{2^n}\to \C^{2^m}$. Hence $\interp{\cdot}:\pW\to \cat{Qubit}$ is a pro morphism.

We add the two following generators: the black spider and the binary white spider, whose interpretations are detailed in the next sub-sections.

\subsection{Black Spider}
To manipulate binary words $\alpha\in \{0,1\}^n $ and $\beta\in \{0,1\}^m $, we will denote $\alpha \oplus \beta \in \{0,1\}^n $ the bitwise XOR (if $n=m$), $\alpha\cdot \beta \in \{0,1\}^{n+m}$ the concatenation, $|\alpha|\in \{0,...,n\}$ the Hamming weight, \textit{i.e.}, the number of ones in the word $\alpha$, and $|\alpha|_2 \in \{0,1\}$ the parity of this weight, $0$ if even and $1$ if odd. 
The \emph{black spider} (or black node) is given by the following interpretation:
\begin{center}
  $\displaystyle
  \interp{\def\fig{def}{\setlength{\fboxsep}{0pt}\colorbox{gray!15}{~\strut\input{./figures/\fig/\fig_black.tikz}}}} := \sum_{\substack{u\in\{0,1\}^m\\ v\in\{0,1\}^n\\ |uv| = 1}} \ketbra{u}{v}
  $
\end{center}
In other words, the black spiders gives an output $1$ if and only if exactly one of its legs (either input or outputs) has value $\ket{1}$ and all the others $\ket{0}$.
As inputs and outputs behave exactly the same, one can use cup and caps in order to transform inputs into outputs and vice-versa:
\begin{center}
  $\def\fig{def}{\setlength{\fboxsep}{0pt}\colorbox{gray!15}{~\strut\input{./figures/\fig/\fig_in-out-1.tikz}}} = {\setlength{\fboxsep}{0pt}\colorbox{gray!15}{~\strut\input{./figures/\fig/\fig_in-out-2.tikz}}}$
\end{center}
Moreover, as input order do not matter, one can bend the wires and move black spiders around, without altering the resulting linear map, we say that the black nodes are flexsymmetric \cite{DBLP:phd/hal/Carette21}. 
Flexsymmetry of the black spider allows us to see diagrams as graphs with fixed inputs and outputs edges. Fixing the input and outputs edges, any graph isomorphism preserves the semantics.

With this graphical interpretation in mind, one can understand the interpretation of a scalar diagram, composed of only black spiders, as counting the number of perfect matchings in the underlying graph.
To see this, one can use the interpretation of a single edge, which simply is the identity $\ketbra0+\ketbra1$. 
This interpretation gives a useful insight in the diagrams: given an edge, one can partition the set of perfect matchings between those that have this edge and those that don't:
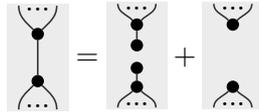
\begin{center}
  $\def\fig{def_edge_rewrite}{\setlength{\fboxsep}{0pt}\colorbox{gray!15}{~\strut\input{./figures/\fig/\fig_0.tikz}}} = {\setlength{\fboxsep}{0pt}\colorbox{gray!15}{~\strut\input{./figures/\fig/\fig_1.tikz}}} + {\setlength{\fboxsep}{0pt}\colorbox{gray!15}{~\strut\input{./figures/\fig/\fig_2.tikz}}}$
\end{center}

In the case where the graph is an actual graph, without half edges, the resulting map is a scalar (no input or outputs).
One can show by induction that this scalar corresponds to the number of ways of choosing a set of edges such that each vertex is covered by exactly one edge. In other ways, \emph{the number of perfect matchings} of the graph.

\subsection{Binary White Spider}
The last generator of the planar W-calculus is the \emph{binary white spider}, given, for any $r\in \C$, by:
\begin{center}
  $\displaystyle
  \interp{\def\fig{def}{\setlength{\fboxsep}{0pt}\colorbox{gray!15}{~\strut\input{./figures/\fig/\fig_white.tikz}}}} :=
  \begin{pmatrix}
    1 & 0\\
    0 & r
  \end{pmatrix}
  $
\end{center}
which corresponds to the usual binary white spider with weight $r$ of the ZW-calculus.
This binary spider corresponds to having a weight $r$ on an edge of the graph. When $r\in\N$, the interpretation is straightforward: the white spider can be replaced by $r$ edges:
  ${\def\fig{def}{\setlength{\fboxsep}{0pt}\colorbox{gray!15}{~\strut\input{./figures/\fig/\fig_white_N_0.tikz}}} = \def\fig{def}{\setlength{\fboxsep}{0pt}\colorbox{gray!15}{~\strut\input{./figures/\fig/\fig_white_N.tikz}}}}$.
And in particular, 
$\def\fig{def}{\setlength{\fboxsep}{0pt}\colorbox{gray!15}{~\strut\input{./figures/\fig/\fig_white_1_0.tikz}}} = \def\fig{def}{\setlength{\fboxsep}{0pt}\colorbox{gray!15}{~\strut\input{./figures/\fig/\fig_white_1_1.tikz}}}$.

Let us interpret the white spiders as weights on the edges of a planar graph $G$ with black spiders on their vertices.
Consider one perfect matching of the same graph $G'$ without weights and consider one perfect matching $P$ of $G'$. If the edge $e$ that belongs to $P$ has a weight $r\in\N$, then it can be replaced by $r$ edges. In other words, the single perfect matching $P$ is replaced by $r$ perfect matchings when $e$ has weight $r$.
By doing this for every edges, one can see that each perfect matching in $G'$ corresponds to a perfect matching of $G$ with a \emph{weight} that is the product of all its edge weights, instead of weight $1$ in $G'$.
For $r\in\C$, one cannot replace a white spider by a given number of edges, but the interpretation is the same: the edge contribute to the perfect matchings that contain it with a \emph{weight} $r$.

\begin{example}
  $\def\fig{def_example}{\setlength{\fboxsep}{0pt}\colorbox{gray!15}{~\strut\input{./figures/\fig/\fig_0.tikz}}} = {\setlength{\fboxsep}{0pt}\colorbox{gray!15}{~\strut\input{./figures/\fig/\fig_1.tikz}}} + {\setlength{\fboxsep}{0pt}\colorbox{gray!15}{~\strut\input{./figures/\fig/\fig_2.tikz}}} = 2 - 1 = 1$
\end{example}

Diagrams generated by the black and binary white node, within the framework described at the beginning of the section, are called $\pW$-diagrams.

\subsection{The FKT Algorithm}
In general, counting the number of perfect matchings in a graph is an \textsc{\#P}-complete problem \cite{Valiant}. However, for planar graphs the same problem turns out to be surprisingly easy, as Fisher, Temperley and Kastelyn showed that it is in \textsc{P} \cite{doi:10.1080/14786436108243366, kasteleyn1967graph}.
The main idea behind the algorithm is that for planar graphs, it is possible to find a good orientation of the edges (called a Pfaffian orientation) in polynomial time such that the number of perfect matchings is the Pfaffian of the adjacency matrix $A$ (actually its skew-symmetric version, called Tutte matrix) of the oriented graph. A result due to Cayley then shows that the Pfaffian is the square root of the determinant of $A$.

Note that one can find such an orientation for any planar graph, even weighted with complex weights, and the equality $pf(A) = \sqrt{det(A))}$ still holds. Therefore, computing the total \emph{weight} of perfect matchings in a complex-weighted graph is in \textsc{P}.

\begin{proposition}\label{prop:P}
  Let $D$ be a scalar $\pW$-diagram. Then $\interp D$ is computable in polynomial time in the number of black nodes.
\end{proposition}

\subsection{Fermionic Swap}
The usual ZW-calculus does have another generator that we did not explicitly include in our fragment, called the \emph{fermionic swap}:
\begin{center}
  $\displaystyle\interp{\tikzfig{fswap}} :=
  \sum_{x,y\in\{0,1\}}(-1)^{xy}\ketbra{x}{y}
  $
\end{center}
 However, it turns out that the fermionic swap is just syntactic sugar, and it is actually in our fragment:
\def\fig{fswap-decomposition}
\begin{align*}
{\setlength{\fboxsep}{0pt}\colorbox{gray!15}{~\strut\input{./figures/\fig/\fig_00.tikz}}}
~:=~{\setlength{\fboxsep}{0pt}\colorbox{gray!15}{~\strut\input{./figures/\fig/\fig_01.tikz}}}
\end{align*}

Notice that the previous equation also appears in \cite{DBLP:journals/toc/CaiG14} to relate planar and non-planar matchgates. It is very useful to treat this piece of diagram as a generator of its own, especially as a particular kind of swap, which shares a lot of (but not all) properties of the symmetric braiding of props. In particular:
		\def\fig{fswap-invol-prf}
		\begin{align*}
			{\setlength{\fboxsep}{0pt}\colorbox{gray!15}{~\strut\input{./figures/\fig/\fig_00.tikz}}}
			\eq{}{\setlength{\fboxsep}{0pt}\colorbox{gray!15}{~\strut\input{./figures/\fig/\fig_05.tikz}}}
			\hspace*{4em}
			\def\fig{fpswap}
			{\setlength{\fboxsep}{0pt}\colorbox{gray!15}{~\strut\input{./figures/\fig/\fig_00.tikz}}}
			\eq{}(-1)^{|D|}{\setlength{\fboxsep}{0pt}\colorbox{gray!15}{~\strut\input{./figures/\fig/\fig_01.tikz}}}
		\end{align*}
%

Where $|D|$ is the number of black nodes in the diagram $D$.

\section{Completeness}\label{sec:3}

The planar W-calculus is introduced with an equational theory, given in \Cref{fig:axioms}, relating together diagrams with the same semantics. We write $\pW\vdash D_1=D_2$ when one can turn diagram $D_1$ into diagram $D_2$ by applying the equations of \Cref{fig:axioms} locally.

\begin{figure}[!htb]
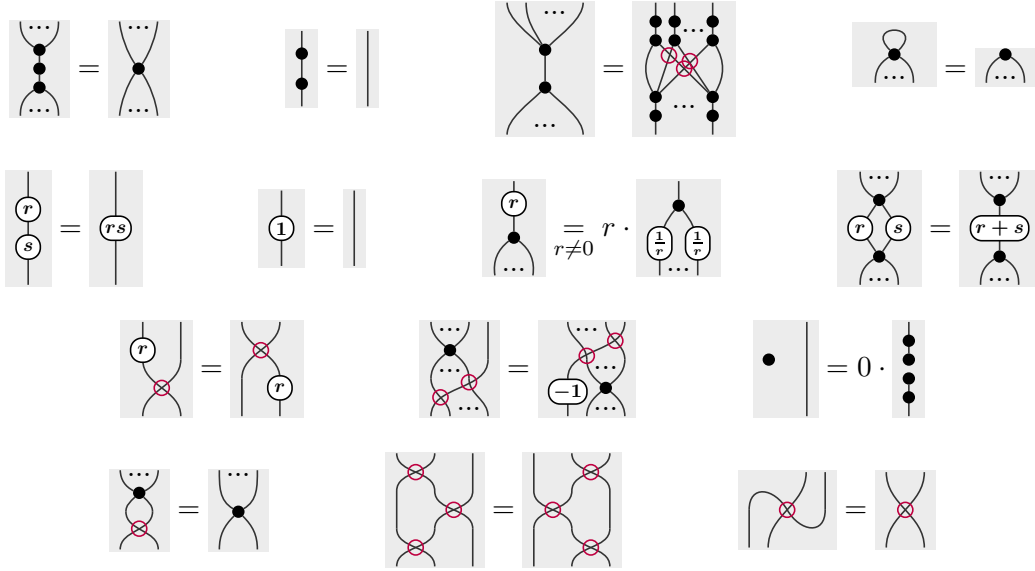

	\centering
	$\def\fig{ax-W-spider}
	{\setlength{\fboxsep}{0pt}\colorbox{gray!15}{~\strut\input{./figures/\fig/\fig_00.tikz}}}
	\eq{}{\setlength{\fboxsep}{0pt}\colorbox{gray!15}{~\strut\input{./figures/\fig/\fig_01.tikz}}}
	\hspace*{4em}
	\def\fig{ax-W-binary}
	{\setlength{\fboxsep}{0pt}\colorbox{gray!15}{~\strut\input{./figures/\fig/\fig_00.tikz}}}
	\eq{}{\setlength{\fboxsep}{0pt}\colorbox{gray!15}{~\strut\input{./figures/\fig/\fig_01.tikz}}}
	\hspace*{4em}
	\def\fig{ax-W-bialgebra}
	{\setlength{\fboxsep}{0pt}\colorbox{gray!15}{~\strut\input{./figures/\fig/\fig_00.tikz}}}
	\eq{}{\setlength{\fboxsep}{0pt}\colorbox{gray!15}{~\strut\input{./figures/\fig/\fig_01.tikz}}}
	\hspace*{4em}
	\def\fig{ax-W-loop}
	{\setlength{\fboxsep}{0pt}\colorbox{gray!15}{~\strut\input{./figures/\fig/\fig_00.tikz}}}
	\eq{}{\setlength{\fboxsep}{0pt}\colorbox{gray!15}{~\strut\input{./figures/\fig/\fig_01.tikz}}}
	$
	
	\bigskip
	$
	\def\fig{ax-phase-fusion}
	{\setlength{\fboxsep}{0pt}\colorbox{gray!15}{~\strut\input{./figures/\fig/\fig_00.tikz}}}
	\eq{}{\setlength{\fboxsep}{0pt}\colorbox{gray!15}{~\strut\input{./figures/\fig/\fig_01.tikz}}}
	\hspace*{4em}
	\def\fig{ax-Z-binary}
	{\setlength{\fboxsep}{0pt}\colorbox{gray!15}{~\strut\input{./figures/\fig/\fig_00.tikz}}}
	\eq{}{\setlength{\fboxsep}{0pt}\colorbox{gray!15}{~\strut\input{./figures/\fig/\fig_01.tikz}}}
	\hspace*{4em}
	\def\fig{ax-phase-distrib}
	{\setlength{\fboxsep}{0pt}\colorbox{gray!15}{~\strut\input{./figures/\fig/\fig_00.tikz}}}
	\eq{r\neq0}r\cdot{\setlength{\fboxsep}{0pt}\colorbox{gray!15}{~\strut\input{./figures/\fig/\fig_01.tikz}}}
	\hspace*{4em}
	\def\fig{ax-sum}
	{\setlength{\fboxsep}{0pt}\colorbox{gray!15}{~\strut\input{./figures/\fig/\fig_00.tikz}}}
	\eq{}{\setlength{\fboxsep}{0pt}\colorbox{gray!15}{~\strut\input{./figures/\fig/\fig_01.tikz}}}
	$
	
	\bigskip
	$
	\def\fig{ax-fswap-Z}
	{\setlength{\fboxsep}{0pt}\colorbox{gray!15}{~\strut\input{./figures/\fig/\fig_00.tikz}}}
	\eq{}{\setlength{\fboxsep}{0pt}\colorbox{gray!15}{~\strut\input{./figures/\fig/\fig_01.tikz}}}
	\hspace*{4em}
	\def\fig{ax-fswaps-W}
	{\setlength{\fboxsep}{0pt}\colorbox{gray!15}{~\strut\input{./figures/\fig/\fig_00.tikz}}}
	\eq{}{\setlength{\fboxsep}{0pt}\colorbox{gray!15}{~\strut\input{./figures/\fig/\fig_01.tikz}}}
	\hspace*{4em}
	\def\fig{ax-zero}
	{\setlength{\fboxsep}{0pt}\colorbox{gray!15}{~\strut\input{./figures/\fig/\fig_00.tikz}}}
	\eq{}0\cdot{\setlength{\fboxsep}{0pt}\colorbox{gray!15}{~\strut\input{./figures/\fig/\fig_01.tikz}}}
	$
	
	\bigskip
	$
	\def\fig{ax-fswap-removal}
	{\setlength{\fboxsep}{0pt}\colorbox{gray!15}{~\strut\input{./figures/\fig/\fig_00.tikz}}}
	\eq{}{\setlength{\fboxsep}{0pt}\colorbox{gray!15}{~\strut\input{./figures/\fig/\fig_01.tikz}}}
	\hspace*{4em}
	\def\fig{ax-fswap-YB}
	{\setlength{\fboxsep}{0pt}\colorbox{gray!15}{~\strut\input{./figures/\fig/\fig_00.tikz}}}
	\eq{}{\setlength{\fboxsep}{0pt}\colorbox{gray!15}{~\strut\input{./figures/\fig/\fig_01.tikz}}}
	\hspace*{4em}
	\def\fig{ax-fswap-rotated}
	{\setlength{\fboxsep}{0pt}\colorbox{gray!15}{~\strut\input{./figures/\fig/\fig_00.tikz}}}
	\eq{}{\setlength{\fboxsep}{0pt}\colorbox{gray!15}{~\strut\input{./figures/\fig/\fig_01.tikz}}}
	$
	\caption{Axioms of the planar W-calculus.}
	\label{fig:axioms}
\end{figure}

\begin{proposition}
	The equational theory of \Cref{fig:axioms} preserves the semantics:
	\[\pW\vdash D_1=D_2 \quad\implies\quad \interp{D_1}=\interp{D_2}\]
\end{proposition}

In the following, we will show that the converse also holds, that is, that whenever two diagrams have the same semantics, they can be turned into one another using the equational theory. Intuitively, this implies that the equational theory completely captures the interaction of generators with one another in the fragment.

To show this result, we give a notion of normal form, which we call W-graph-state with X-gates (WGS-X for short), then a refinement of that normal form (reduced WGS-X form) which can be shown to be unique, and we give a rewrite strategy (derivable from the equational theory) to turn any $\pW$-diagram into this form.

\subsection{Normal Form}

The first step we take towards defining a normal form is a simplification, making use of the compact structure of the underlying pro, where we relate maps and states:
\begin{proposition}
	\label{prop:process-state-duality}
	There is an isomorphism between $\pW(n,m)$ and $\pW(0,n+m)$ defined as such:
	\def\fig{process-state-duality}
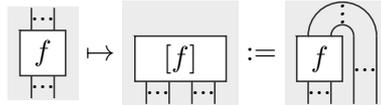
	\begin{align*}
		{\setlength{\fboxsep}{0pt}\colorbox{gray!15}{~\strut\input{./figures/\fig/\fig_00.tikz}}}
		\def\fig{sateform}
		\mapsto{\setlength{\fboxsep}{0pt}\colorbox{gray!15}{~\strut\input{./figures/\fig/\fig_00.tikz}}}
		\def\fig{process-state-duality}
		:={\setlength{\fboxsep}{0pt}\colorbox{gray!15}{~\strut\input{./figures/\fig/\fig_01.tikz}}}
	\end{align*}
\end{proposition}
This isomorphism allows us only to consider states rather than maps in the following.

Then, we define W-graph-states, by first defining ordered weighted graphs:
\begin{definition}[Ordered $R$-Weighted Graph]
	$G=(V,E,w)$ is called an ordered $R$-weigthed graph if:
	\begin{itemize}
		\item $V$ is a set endowed with a total order $\prec$ (or equivalently a sequence)
		\item $E \subset V\times V$ is such that $(u,v)\in E\implies u\prec v$
		\item $w:E\to R\setminus\{0\}$ maps each edge to its weight
	\end{itemize}
\end{definition}

\begin{definition}[W-Graph-State]
	Let $G=(V,E,w)$ be an ordered weighted graph. Then, $\operatorname{WGS}(G)$ is defined as the $\pW$-diagram where:
	\begin{itemize}
		\item Each vertex in $V$ gives a W-spider linked to an output through an additional $\tikzfig{X}$ (the order on $V$ gives the order of the outputs)
		\item Each (weighted) edge $(u,v)$ gives a white dot with parameter $w((u,v))$ linked to the W-spiders obtained from $u$ and $v$
		\item All wire crossings in $\operatorname{WGS}(G)$ are fermionic swaps $\tikzfig{fswap}$
		\item No output wire crosses another wire
		\item There are no self-intersecting wires
	\end{itemize}
\end{definition}
When an edge has weight $1$ we may ignore the white dot and represent the edge as a simple wire, since $\def\fig{ax-Z-binary}
{\setlength{\fboxsep}{0pt}\colorbox{gray!15}{~\strut\input{./figures/\fig/\fig_00.tikz}}}
\eq{}{\setlength{\fboxsep}{0pt}\colorbox{gray!15}{~\strut\input{./figures/\fig/\fig_01.tikz}}}$. 
Notice that there are several ways to build $\operatorname{WGS}(G)$, but all of them are equivalent thanks to $\def\fig{ax-fswap-Z}
{\setlength{\fboxsep}{0pt}\colorbox{gray!15}{~\strut\input{./figures/\fig/\fig_00.tikz}}}
\eq{}{\setlength{\fboxsep}{0pt}\colorbox{gray!15}{~\strut\input{./figures/\fig/\fig_01.tikz}}}$ and the axioms on the fermionic swap $\tikzfig{fswap}$, together with the provable identities in Lemmas \ref{lem:floop} and \ref{lem:fswap-invol}:

\begin{multicols}{2}
	
	\begin{lemma}
		\label{lem:floop}
		\def\fig{floop-prf}
		\begin{align*}
			{\setlength{\fboxsep}{0pt}\colorbox{gray!15}{~\strut\input{./figures/\fig/\fig_00.tikz}}}
			\eq{}{\setlength{\fboxsep}{0pt}\colorbox{gray!15}{~\strut\input{./figures/\fig/\fig_05.tikz}}}
		\end{align*}
	\end{lemma}
	
	\begin{lemma}
		\label{lem:fswap-invol}
		\def\fig{fswap-invol-prf}
		\begin{align*}
			{\setlength{\fboxsep}{0pt}\colorbox{gray!15}{~\strut\input{./figures/\fig/\fig_00.tikz}}}
			\eq{}{\setlength{\fboxsep}{0pt}\colorbox{gray!15}{~\strut\input{./figures/\fig/\fig_05.tikz}}}
		\end{align*}
	\end{lemma}
	
\end{multicols}

\begin{definition}[WGS-X form]
	We say that a $\pW$-state $D$ on $n$ qubits is in:
	\begin{itemize}
		\item \textbf{WGS-X form} if there exist $s\in \mathbb C$, $G=([1,n],E,w)$ an ordered graph, and $\vec b \in\{0,1\}^n$ such that $D = s\cdot\left(\tikzfig{X}^{\otimes \vec b}\right)\circ \operatorname{WGS}(G)$.
		\item \textbf{pseudo-WGS-X form} if it is in WGS-X form with potentially vertices linked to several outputs, additional $\tikzfig{edge}$ ($r\neq0$) on wires that do not correspond to edges in the graph, and potentially fermionic swaps $\tikzfig{fswap}$ between outputs.
		\item \textbf{reduced WGS-X form} (rWGS-X) if it is in WGS-X form and:
		\[\forall i,~\left(b_i=0\implies \nexists j,~ (i,j)\in E\right)\]
		i.e.~$b_i=0$ is only possible if vertex $i$ has no neighbour on its right.
	\end{itemize}
\end{definition}

\begin{example}
	\def\fig{example-graph-state}
	$\operatorname{WGS}\left( \input{./figures/example-graph-state/example-graph-state_00.tikz} \right) = {\setlength{\fboxsep}{0pt}\colorbox{gray!15}{~\strut\input{./figures/\fig/\fig_01.tikz}}}$
	where in the starting graph, vertices are ordered left to right, and edges with no indication of weight have weight $1$.
	
	If $\vec b = (0,1,1,0,1)$, then the obtained WGS-X state is:
	\[s\cdot{\setlength{\fboxsep}{0pt}\colorbox{gray!15}{~\strut\input{./figures/\fig/\fig_02.tikz}}}\]
	where we used the fact that $\tikzfig{X}$ is an involution to simplify the diagram. 
	The WGS-X state is however not reduced, as both the first and fourth qubits have additional $\tikzfig{X}$ applied to them, but still have neighbours on their right.
	
	Finally, the following diagram is an example of a pseudo-WGS-X state: 
	\[{\setlength{\fboxsep}{0pt}\colorbox{gray!15}{~\strut\input{./figures/\fig/\fig_03.tikz}}}\]
	
\end{example}

\subsection{Rewrite Strategy}

We define in this section a rewrite strategy, derived from the equational theory, that will terminate in a normal form (WGS-X). Doing this naively is made difficult by the potential presence of fermionic swaps $\tikzfig{fswap}$ wherever we are looking for patterns to rewrite. Thankfully, the last 5 equations in \Cref{fig:axioms},together with the above Lemmas \ref{lem:floop} and \ref{lem:fswap-invol} essentially tell us that we can treat those as usual swaps with the only catch that removing self loops or moving wires past black nodes adds a $-1$ weight to the wires.

In the upcoming rewrite strategy, we will hence only specify the patterns without potential fermionic swaps inside. Should there be some present, it is understood that they will be moved out of the pattern, before the rewrite occurs. The rules necessary for the rewrite strategy are given in \Cref{fig:rewrites}.

%

\begin{figure}[!htb]
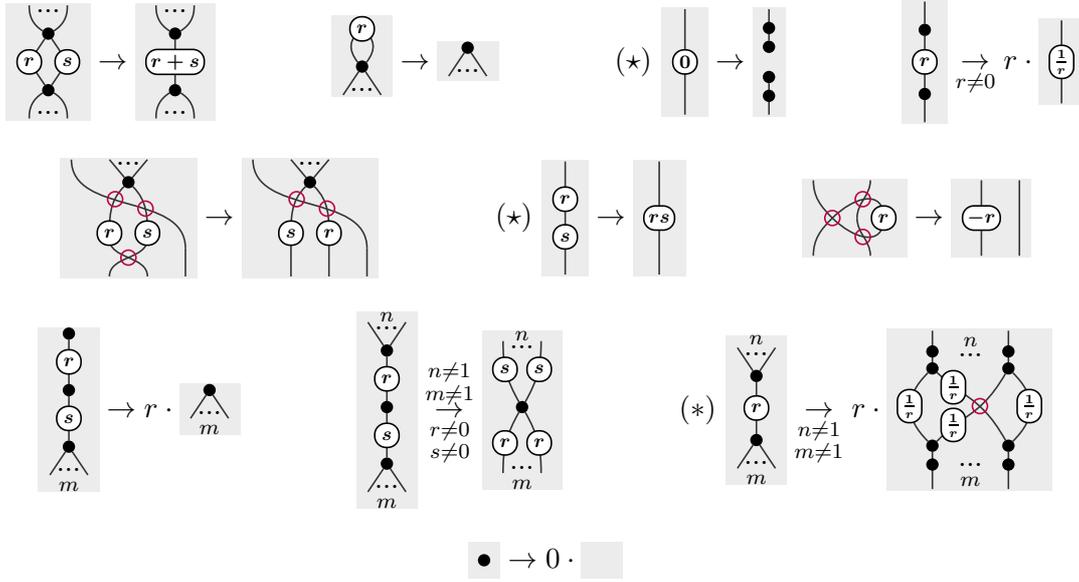

	\centering
	$
	\def\fig{ax-sum}
	{\setlength{\fboxsep}{0pt}\colorbox{gray!15}{~\strut\input{./figures/\fig/\fig_00.tikz}}}
	\to{\setlength{\fboxsep}{0pt}\colorbox{gray!15}{~\strut\input{./figures/\fig/\fig_01.tikz}}}
	\hspace*{4em}
	\def\fig{rewrite-W-loop}
	{\setlength{\fboxsep}{0pt}\colorbox{gray!15}{~\strut\input{./figures/\fig/\fig_00.tikz}}}
	\to{\setlength{\fboxsep}{0pt}\colorbox{gray!15}{~\strut\input{./figures/\fig/\fig_05.tikz}}}
	\hspace*{4em}
	\def\fig{0-edge-prf}
	(\star)~
	{\setlength{\fboxsep}{0pt}\colorbox{gray!15}{~\strut\input{./figures/\fig/\fig_00.tikz}}}
	\to{\setlength{\fboxsep}{0pt}\colorbox{gray!15}{~\strut\input{./figures/\fig/\fig_07.tikz}}}
	\hspace*{4em}
	\def\fig{rewrite-binary-W}
	{\setlength{\fboxsep}{0pt}\colorbox{gray!15}{~\strut\input{./figures/\fig/\fig_00.tikz}}}
	\underset{r\neq0}\to r\cdot{\setlength{\fboxsep}{0pt}\colorbox{gray!15}{~\strut\input{./figures/\fig/\fig_02.tikz}}}
	$
	
	\bigskip
	$
	\def\fig{rewrite-fswap-removal}
	{\setlength{\fboxsep}{0pt}\colorbox{gray!15}{~\strut\input{./figures/\fig/\fig_00.tikz}}}
	\to{\setlength{\fboxsep}{0pt}\colorbox{gray!15}{~\strut\input{./figures/\fig/\fig_01.tikz}}}
	\hspace*{4em}
	\def\fig{ax-phase-fusion}
	(\star)~
	{\setlength{\fboxsep}{0pt}\colorbox{gray!15}{~\strut\input{./figures/\fig/\fig_00.tikz}}}
	\to{\setlength{\fboxsep}{0pt}\colorbox{gray!15}{~\strut\input{./figures/\fig/\fig_01.tikz}}}
	\hspace*{4em}
	\def\fig{rewrite-floop}
	{\setlength{\fboxsep}{0pt}\colorbox{gray!15}{~\strut\input{./figures/\fig/\fig_00.tikz}}}
	\to{\setlength{\fboxsep}{0pt}\colorbox{gray!15}{~\strut\input{./figures/\fig/\fig_01.tikz}}}
	$
	
	\bigskip
	$
	\def\fig{rewrite-fusion-0}
	{\setlength{\fboxsep}{0pt}\colorbox{gray!15}{~\strut\input{./figures/\fig/\fig_00.tikz}}}
	\to r\cdot{\setlength{\fboxsep}{0pt}\colorbox{gray!15}{~\strut\input{./figures/\fig/\fig_03.tikz}}}
	\hspace*{4em}
	\def\fig{rewrite-fusion}
	{\setlength{\fboxsep}{0pt}\colorbox{gray!15}{~\strut\input{./figures/\fig/\fig_00.tikz}}}
	\overset{\substack{n\neq1\\m\neq1}}{\underset{\substack{r\neq0\\s\neq0}}\to}{\setlength{\fboxsep}{0pt}\colorbox{gray!15}{~\strut\input{./figures/\fig/\fig_03.tikz}}}
	\hspace*{4em}
	\def\fig{rewrite-pivot}
	(\ast)~{\setlength{\fboxsep}{0pt}\colorbox{gray!15}{~\strut\input{./figures/\fig/\fig_00.tikz}}}
	\underset{\substack{n\neq1\\m\neq1}}\to r\cdot{\setlength{\fboxsep}{0pt}\colorbox{gray!15}{~\strut\input{./figures/\fig/\fig_03.tikz}}}$
	
	\bigskip
	$
	\def\fig{rewrite-zero}
	{\setlength{\fboxsep}{0pt}\colorbox{gray!15}{~\strut\input{./figures/\fig/\fig_00.tikz}}}
	\to0\cdot{\setlength{\fboxsep}{0pt}\colorbox{gray!15}{~\strut\input{./figures/\fig/\fig_01.tikz}}}$
	\caption{Rewrite rules. All these rules except $(\star)$ are supposed to apply when any of the white nodes are replaced by identity (i.e.~when their weight is $1$). Rule $(\ast)$ can only be applied if at least one of the black nodes is internal, and if none of the other rules applies.}
	\label{fig:rewrites}
\end{figure}

\begin{figure}[!htb]
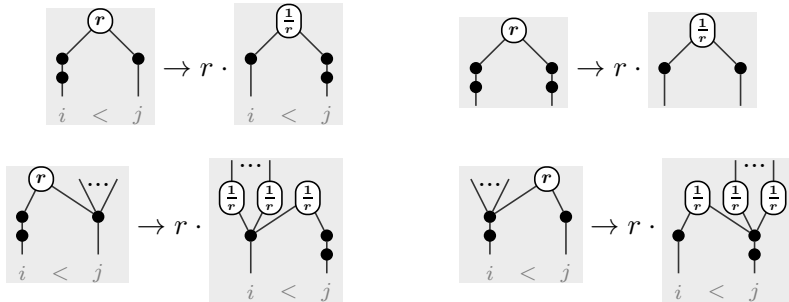

	\centering
	\def\fig{reducing-WGS-X}
	$
	{\setlength{\fboxsep}{0pt}\colorbox{gray!15}{~\strut\input{./figures/\fig/\fig_00.tikz}}}
	\to r\cdot{\setlength{\fboxsep}{0pt}\colorbox{gray!15}{~\strut\input{./figures/\fig/\fig_01.tikz}}}
	\hspace*{4em}
	{\setlength{\fboxsep}{0pt}\colorbox{gray!15}{~\strut\input{./figures/\fig/\fig_02.tikz}}}
	\to r\cdot{\setlength{\fboxsep}{0pt}\colorbox{gray!15}{~\strut\input{./figures/\fig/\fig_03.tikz}}}
	$
	
	\bigskip
	$
	{\setlength{\fboxsep}{0pt}\colorbox{gray!15}{~\strut\input{./figures/\fig/\fig_04.tikz}}}
	\to r\cdot{\setlength{\fboxsep}{0pt}\colorbox{gray!15}{~\strut\input{./figures/\fig/\fig_05.tikz}}}
	\hspace*{4em}
	{\setlength{\fboxsep}{0pt}\colorbox{gray!15}{~\strut\input{./figures/\fig/\fig_06.tikz}}}
	\to r\cdot{\setlength{\fboxsep}{0pt}\colorbox{gray!15}{~\strut\input{./figures/\fig/\fig_07.tikz}}}$
	\caption{Rules for reduced WGS-X form, together with rule $(\ast)$ when the leftmost black node is a type-0 boundary node.}
	\label{fig:rewrite-reduce}
\end{figure}

\begin{proposition}
	\label{prop:rewrites-derivable}
	The rewrite rules of \Cref{fig:rewrites} are derivable from the equational theory of \Cref{fig:axioms} and hence are sound.
\end{proposition}

For the rewrite strategy to terminate, we need to distinguish between different types of nodes:
\begin{definition}[Boundary Node / Internal Node]
	A node is a \emph{boundary node of type 1} if it is linked directly to an output. A node is a \emph{boundary node of type 0} if it is connected to a binary boundary node of type 1.\\
	We say that a black node of $D$ is \emph{internal} if it is not a boundary node.
\end{definition}
The rewrite strategy is then laid out as follows:
\begin{definition}[Rewrite Strategy]
\label{def:rewrite-strat}
The rewrite strategy is defined in 3 steps:
\begin{enumerate}
\item \label{item:pseudo-WGS-X} Apply the rewrites of \Cref{fig:rewrites} in any order but following constraints, until none apply anymore. The diagram ends up in pseudo-WGS-X form.
\item \label{item:WGS-X} 
First, whenever a type-1 boundary is linked to $n>1$ outputs directly, apply $\def\fig{ax-W-binary}
{\setlength{\fboxsep}{0pt}\colorbox{gray!15}{~\strut\input{./figures/\fig/\fig_01.tikz}}}
\to{\setlength{\fboxsep}{0pt}\colorbox{gray!15}{~\strut\input{./figures/\fig/\fig_00.tikz}}}$
to the $n-1$ rightmost such outputs (the top black node then becomes a type-0 boundary node, the bottom one a type-1 boundary node). 
Then, push all potential fermionic swaps $\tikzfig{fswap}$ between outputs inside the graph part. 
Finally, move boundary weights up into the edges of the WGS using $\def\fig{ax-phase-distrib}{\setlength{\fboxsep}{0pt}\colorbox{gray!15}{~\strut\input{./figures/\fig/\fig_00.tikz}}}\to r\cdot{\setlength{\fboxsep}{0pt}\colorbox{gray!15}{~\strut\input{./figures/\fig/\fig_01.tikz}}}$.
The diagram ends up in WGS-X form.
\item \label{item:rWGS-X} Whenever a type-0 vertex in the graph has a right neighbour, depending on the arity of the nodes, apply rule $(\ast)$ or one of the rules of \cref{fig:rewrite-reduce} between the two nodes (and apply any other possible rule before going on).
\end{enumerate}
\end{definition}

A claim was made in \Cref{def:rewrite-strat} about the form of the diagram at the end of each step. Those claims are going to be proven in the following (\Cref{lem:rewrite-to-WGS-X}). At the same time, we are going to show that the rewrite terminates.

\begin{proposition}[Termination in rWGS-X form]
	\label{prop:termination}
	\label{lem:rewrite-to-WGS-X}
	The rewrite strategy terminates in polynomial time. Moreover, after Step \ref{item:pseudo-WGS-X} of the rewrite, the diagram is indeed in pseudo-WGS-X form, after Step \ref{item:WGS-X}, it is in WGS-X form, and after Step \ref{item:rWGS-X}, it is in rWGS-X form.
\end{proposition}



An important operation on WGS-X states that has a simple graphical interpretation is the following:
\begin{lemma}
	\label{lem:vertex-removal}
	For any diagram $D$ in WGS-X form $(s,G,\vec b)$, applying $\tikzfig{bra1}\circ\tikzfig{X}^{\otimes(b_i\oplus1)}$ on the $i$th output can be turned into the WGS-X form $(s,G\setminus\{i\}, \vec b\setminus b_i)$, where $G\setminus\{i\}$ is defined as the graph $G$ from which vertex $i$ is removed (together with all edges linked to $i$ and their weights), and similarly $\vec b\setminus b_i$ is defined as the sequence $\vec b$ from which $i$th element is removed.
\end{lemma}

This allows us to prove the following:

\begin{lemma}
	\label{lem:WGS-X-zero}
	For any diagram $D$ in WGS-X form $(s,G,\vec b)$:
	\[\interp{D}=0\iff s=0\]
\end{lemma}
We may then prove that $0$-diagrams can be put in a very well-defined form:
\begin{lemma}
	\label{lem:WGS-X-zero-equivalence}
	Let $D$ be a WGS-X state such that $\interp{D}=0$. Then $D$ can be put in the WGS-X form $(0, G=([1,n], \emptyset, \_), \vec 0)$, i.e.:
	\[\pW\vdash D = 0\cdot\tikzfig{ket0s}\] 
\end{lemma}

We are now able to prove the completeness of the equational theory.
\begin{theorem}
	\label{thm:completeness}
	Let $D_1$ and $D_2$ be two $\pW$-diagrams. Then:
	\[\interp{D_1}=\interp{D_2} \iff \pW\vdash D_1=D_2\]
\end{theorem}

This last theorem, together with the fact that the rewriting in rWGS-X form is polynomial (\Cref{prop:termination}) makes the problem of deciding whether two $\pW$-diagrams are semantically equivalent a \textsc{P} problem.

\section{Matchgates}\label{sec:4}

This section aims at characterising exactly the linear maps that $W$-diagrams represent.

\subsection{Matchgate Identities}

Valiant first introduced matchgate identities to characterise $2\to 2$ matchgates, a family of linear maps described in a combinatorial way \cite{DBLP:journals/siamcomp/Valiant02}. In \cite{DBLP:journals/mst/CaiCL09}, the matchgate identities have been extended to characterise matchgates of any size. In the literature, there is a close link between matchgate identities and the Grassman-Plucker identities applied to Pfaffians. It is not the case here, as the diagrammatic technics allow us to directly link matchgate identities to matchings without the intermediate of the Pfaffian. We can fully recover the connection with Pfaffians through the Fetcher-Kasteleyn-Temperley algorithm for counting perfect matchings \cite{kasteleyn1961statistics, doi:10.1080/14786436108243366}, more details on this are outlined in Section \ref{sec:5}.
Many of the proofs of this section are inspired by the very useful clarification of matchgate theory presented in \cite{DBLP:journals/toc/CaiG14}. Notice that contrary to the literature that differentiates between matchgrids, matchcircuits or matchnets, we will only use the term matchgate for any linear map satisfying the matchgate identities.

Recall that for binary words $\alpha\in \{0,1\}^n $ and $\beta\in \{0,1\}^m $, $\alpha \oplus \beta \in \{0,1\}^n $ is the bitwise XOR (if $n=m$), $\alpha\cdot \beta \in \{0,1\}^{n+m}$ the concatenation, $|\alpha|\in \{0,...,n\}$ the Hamming weight, \textit{i.e.}, the number of ones in the word $\alpha$, and $|\alpha|_2 \in \{0,1\}$ the parity of this weight, $0$ if even and $1$ if odd. 

\begin{definition}[Matchgate Identities]
	A tensor $\Gamma \in \mathbb{C}^{2^n}$ satisfies the \textbf{matchgate identities} (MGIs) if for all $\alpha,\beta \in \{0,1\}^n $:
	
	\[\sum_{k=1}^{|\alpha \oplus \beta| } (-1)^k \Gamma_{\alpha \oplus e_{p_k} } \Gamma_{\beta \oplus e_{p_k} }=0 \]

	Where $e_{p_k} \in \{0,1\}^n $ is the binary word which is zero everywhere except in position $p_k $, which is the $k$th position in the set $\{p_1 , ... , p_{|\alpha \oplus \beta |}\} \subseteq \{1,...,n\} $ of positions in which the words $\alpha$ and $\beta$ differs.
\end{definition}

The matchgate identities are not linear, so the set of matchgates is not a subspace of the vector space $\mathbb{C}^{2^n}$ but an algebraic variety \cite{DBLP:journals/mst/CaiCL09}. In general, those identities are not algebraically independent, \textit{i.e.} are not all strictly necessary to describe match-tensors.

Indeed, there are numerous symmetries in those identities. For example, the case $\alpha = \beta $ directly gives empty sums and exchanging $\alpha$ and $\beta$ gives the same identity. Interestingly, one can replace half of the identities with a parity condition.

\begin{proposition}[Parity condition \cite{DBLP:journals/toc/CaiG14}]\label{prop:parity}
	If $\Gamma$ satisfies the matchgate identities then it satisfies the \textbf{parity condition}: for all $\alpha,\beta \in \{0,1\}^n $, $ |\alpha|_2 \neq |\beta |_2 \Rightarrow \Gamma_\alpha \Gamma_\beta =0$.
\end{proposition}

The parity condition splits match-tensors into two groups, the one with odd parity, such that $|\alpha|$ even implies $\Gamma_\alpha = 0 $, and the one of even parity, such that $|\alpha|$ odd implies $\Gamma_\alpha = 0 $. In particular, the parity condition directly implies that all terms in identities with $|\alpha|_2 \neq |\beta|_2 $ are zero. Notice that the parity condition is not sufficient. We still need matchgate identities in general.

However, the parity condition is sufficient for $n\leq 3$, but not anymore for $n=4$, the original case considered by Valiant \cite{DBLP:journals/siamcomp/Valiant02}. In particular, for $n=0$, the matchgate identities are trivially true; hence they are satisfied by all scalars (processes $0\to 0 $).

\subsection{The Pro of Matchgates}

We will now use the matchgates to define a pro. So far, matchgate identities have been used to characterise vectors seen as tensors, without consideration of inputs and outputs. To apply them to linear maps $f:n\to m $, we will use the state form: $[f]:0\to n+m $ 
described in \Cref{prop:process-state-duality}. 
It allows us to define matchgates.

\begin{definition}[Matchgates]
	A \textbf{matchgate} is a linear map $f:\mathbb{C}^{2^n}\to \mathbb{C}^{2^m}$ such that $[f]$ satisfies the matchgate identities.
\end{definition}

We would like to define a sub-pro of $\textbf{Qubit}$ whose processes are matchgates, however, there are a few properties to check before that. We start by showing stability by the tensor product.

\begin{lemma}
\label{lem:mtensor}
	Given two linear maps $f:a\to b $ and $g:c\to d $ whose state forms $[f] \in \mathbb{C}^{2^{a+b}}$ and $[g] \in \mathbb{C}^{2^{c+d}}$ satisfy the matchgate identities, then $[f\otimes g]\in \mathbb{C}^{2^{a+c+b+d}}$ satisfies the matchgate identities.
\end{lemma}

The next thing to check is stability by composition; this follows from the following result:

\begin{lemma}
\label{lem:mcomp}
	If $\Gamma \in \mathbb{C}^{2^{n+2}}$ satisfies the matchgate identities, then the tensor obtained by contracting two consecutive indices satisfies the matchgate identities.
\end{lemma}

Notice that the consecutive indices assumption is essential here. Without it, we could easily construct the swap gate that does not satisfy the matchgate identities. To be able to contract consecutive indices is enough to show the stability by composition. The idea is to iterate contraction on consecutive indices until we obtain enough cups to use the snake equation, pictorially:
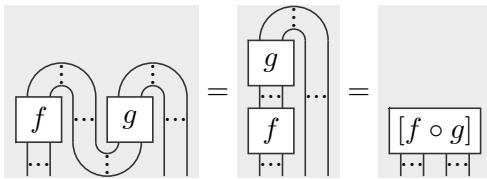
\def\fig{compcontr}
\begin{align*}
	{\setlength{\fboxsep}{0pt}\colorbox{gray!15}{~\strut\input{./figures/\fig/\fig_00.tikz}}}
	\eq{}{\setlength{\fboxsep}{0pt}\colorbox{gray!15}{~\strut\input{./figures/\fig/\fig_01.tikz}}}
	\eq{}{\setlength{\fboxsep}{0pt}\colorbox{gray!15}{~\strut\input{./figures/\fig/\fig_02.tikz}}}
\end{align*}

Now that we have stability by tensor and composition, it only remains to show the identities are matchgates. $id_0 $ is a scalar, so directly a matchgate. The state-form of $id_1 $ is the cap which is a matchgate as it satisfies the parity condition (sufficient for $n=2$). The fact that all $id_n $ are matchgates follows from stability by the tensor product. We can now state the main theorem of this subsection.

\begin{theorem}[\textbf{Match}]
	The matchgates form a pro $\textbf{Match}$, which is a sub-pro of $\textbf{Qubit}$.
\end{theorem}

Notice that $\textbf{Match}$ is compact closed since the cup and the cap are both matchgates. Hence we can also use process/state duality in $\textbf{Match}$ without any worry. As expected, all $W$-diagrams represent matchgates.

\begin{lemma}\label{lem:fact}
	The functor $\interp{\_}: \pW\to \textbf{Qubit} $ factorises through $\textbf{Match}$, \textit{i.e.}, the interpretations of diagrams in $W$ are matchgates.
	\[
		\input{./figures/diagfact.tikz}
	\]

\end{lemma}

\begin{proof}
	We have to prove that the interpretation of any $\pW$ diagram is a matchgate. To do so, as matchgates are stable by composition and tensor product we only have to check that the interpretations of the generators are matchgates. The state forms of the generators have at most three outputs ($n$-ary spiders can be decomposed into binary and ternary spiders), so it is sufficient to check the parity condition, which is indeed satisfied by the interpretations of the generators.
\end{proof}

\subsection{Universality}

Now that we proved that all $\pW$-diagrams represent matchgates, it remains to show that all matchgates can be represented by a $\pW$ diagram, in other words, that $\pW$ is universal for $\textbf{Match}$. This will require a few additional properties of matchgates, adapting some results of \cite{DBLP:journals/toc/CaiG14}. 

\begin{lemma}\label{lem:weight2}
	If $\Gamma$ satisfies the matchgate identities and $\Gamma_{\mathbf{0}} =1 $, where $\mathbf{0}$ is binary word full of $0$, then it is uniquely determined by its coefficients $\Gamma_\alpha$ where $|\alpha | = 2 $.
\end{lemma}

\begin{proof}
	If $|\alpha |=0$ then we already know that $\Gamma_\alpha = 1$ and the parity condition implies that $\Gamma_\alpha = 0 $ if $|\alpha|=1 $. We show that for all $\alpha$ with $ 3\leq|\alpha|$, we can express $\Gamma_\alpha $ from coefficients $\Gamma_\beta $ where all $\beta$s have strictly smaller Hamming weights. Let $i$ be the first position where $\alpha $ and $\mathbf{0}$ differ, the matchgate identity corresponding to $\alpha \oplus e_i $ and $\mathbf{0} \oplus e_i $ is:
	\[\sum_{k=1}^{|\alpha | } (-1)^k \Gamma_{\alpha \oplus e_i \oplus e_{p_k } } \Gamma_{e_i \oplus e_{p_k }}=0\]
	Here the $p_k $ are exactly the position where $\alpha$ is $1$, in particular $i=p_1$ so:
	\[\Gamma_\alpha  = \Gamma_\alpha \Gamma_\mathbf{0} = \sum_{k=2}^{|\alpha | } (-1)^k \Gamma_{\alpha \oplus e_i \oplus e_{p_k } } \Gamma_{e_i \oplus e_{p_k }}\]
	For $k\geq 2$, We have $|e_i \oplus e_{p_k } | = 2 $ and $|\alpha \oplus e_i \oplus e_{p_k }| = |\alpha|-2 $ so $\Gamma_\alpha $ is completely determined by coefficients corresponding to strictly smaller Hamming weight. It follows that all $\Gamma_\alpha $ can be expressed from the $\Gamma_\beta $s with $|\beta|=2$.
	
\end{proof}

We will now be able to reuse the normal form from Section \ref{sec:3} to construct diagrams representing any matchgate.

\begin{lemma}[Universality]
	$\pW$ is universal for $\textbf{Match}$.
\end{lemma}

\begin{proof}
	Relying on process/state duality, we only consider states $0\to n$. Given $\Gamma$ satisfying the matchgate identities, we will construct a $W$ diagram $D$ such that $\interp{D}=\Gamma $. We start by considering the case where $\Gamma_{\mathbf{0}} = 1$. Then we construct a weighted graph $G$ on $n$ vertices setting the weight of the edge $(i,j)$ to $\Gamma_{e_i \oplus e_j }$. We then take $D$ to be the diagram in graph form corresponding to $G$. By construction we then have $\interp{D}_{\mathbf{0}}= 1 $ and $\interp{D}_{ e_i \oplus e_j } = \Gamma_{e_i \oplus e_j} $ for all $i\neq j$. Furthermore, by Lemma \ref{lem:fact}, $\interp{D}$ is a matchgate so by Lemma \ref{lem:weight2}, $\interp{D}=\Gamma $.
	
	Now if $\Gamma_{\mathbf{0}}\neq 1 $: First if $\Gamma_{\mathbf{0}}\neq 0 $ then $ \Gamma' = \frac{1}{\Gamma_{\mathbf{0}}} \Gamma $ is of the right form so we can obtain $D$ by adding a floating edge of weight $\Gamma_{\mathbf{0}}$ to the diagram $D'$ representing $\Gamma' $. The last case is $\Gamma_{\mathbf{0}}= 0 $, then if $\Gamma=0 $ we can represent $\Gamma$ by any diagram and a floating black node, else let $\beta$ be such that $\Gamma_\beta \neq 0 $, then $\Gamma'$ defined as $ \Gamma'_\alpha = \Gamma_{\alpha \oplus \beta }$ satisfies $\Gamma'_\mathbf{0} \neq 0 $ and there is a diagram $D'$ representing $\Gamma' $. A diagram $D$ representing $\Gamma$ is then obtained by plugging binary black nodes to the outputs of $D'$ corresponding to the positions where $\beta $ is $1$.
\end{proof}

Notice that since $\textbf{Match}$ is a sub-pro of $\textbf{Qubit}$, the completeness proof of Section \ref{sec:3} still holds in $\textbf{Match}$. It provides us with a universal and complete graphical language for matchgates.

\begin{theorem}
	$\pW$ is universal and complete for $\textbf{Match}$.
\end{theorem}

\section{Further Work}\label{sec:5}

The proper definition and axiomatisation of the $\pW$-calculus pave the way to diverse investigations of the connection between combinatorics and quantum computing. We briefly outline in this last section some very promising directions that are the subjects of ongoing research. 

\subsection{New Simulation Techniques for Quantum Circuits}

The identification of a fragment of the ZX-calculus exactly corresponding to the efficiently simulable Clifford gate \cite{DBLP:journals/lmcs/BackensPW20} allows to design new rewrite-based simulation technics for quantum circuits introduced in \cite{DBLP:conf/tqc/KissingerWV22}. Those algorithms have a parametrised complexity which is polynomial in the number of Clifford gated but exponential in the number of $T$-gates (a gate outside of the Clifford fragment sufficient to reach approximate universality).

Similarly, we have identified an efficiently simulable fragment of ZW-calculus: the $\pW$-calculus exactly corresponding to matchgates. Adding the swap gate to $\pW$ we obtain another fragment of ZW which is exactly the fermionic ZW-calculus introduced in \cite{DBLP:journals/lmcs/FeliceHN19}. This calculus is universal for \textbf{Qubit} modulo an encoding trick: the dual-rail encoding. Equivalently, LFM is ZW where white nodes are contrived to have even arities, so adding arity one white nodes (corresponding to preparing $\ket{+}$ states) is enough to recover the full ZW-calculus, which is universal for $\textbf{Qubits}$. This situation suggests the possibility of designing rewrite-based simulation algorithms with complexities parametrised by the number of swap gates and/or $\ket{+}$ preparation. It would lead to a brand new kind of quantum simulation technics exploiting the combinatorial structure of matchgate and directly connected to classical perfect matching counting algorithms.

\subsection{Combinatorial Interpretation of Full ZW-Calculus}

In Section \ref{sec:2}, we provided a combinatorial interpretation of $\pW$-diagrams \textit{via} perfect matchings in planar graphs. This combinatorial approach directly extends to LFM-calculus \textit{via} perfect matchings in arbitrary graphs (which is \textsc{\#P}-complete). Furthermore, we can also extend the interpretation to the full ZW-calculus, where white nodes can have arbitrary arities. To do so, we must consider hypergraph matchings, \textit{i.e.}, subsets of hyperedges covering each vertex exactly once. The arbitrary arity white nodes here play the role of hyperedges, and the black nodes, the role of vertices. Thus, the interpretation of ZW-scalars is the number of hypergraph matchings of the hypergraph underlying the diagram. Notice that hypergraph matching is also presented as the set cover problem in the literature. The full ZW-calculus could offer new perspectives on set cover in the same way that $\pW$ does for perfect matchings. In particular, some reduction results appear to have very clear diagrammatical proofs.

Aside from perfect matchings, it seems that graphical languages can encode other counting problems on graphs or hypergraphs. Designing such languages could shed a new tensorial/diagrammatical light on the corresponding combinatorial problems. Those approaches are reminiscent of the recent ZH-based algorithm for \textsc{\#Sat}, introduced in \cite{DBLP:journals/corr/abs-2212-08048} and related works linking graphical languages and counting complexity \cite{DBLP:journals/corr/abs-2004-06455,DBLP:conf/icalp/BeaudrapKW22}. Conversely, the question of applying similar combinatorial interpretations to other graphical languages as ZX-calculus \cite{DBLP:conf/icalp/CoeckeD08}, or ZH-calculus \cite{backens2019zh} is also worth being investigated. 

\subsection{Towards a Diagrammatic Approach of Perfect Matching Counting}

In Section \ref{sec:2}, we used the Fletcher-Kasteleyn-Temperley algorithm as a black box to compute the interpretation of $\pW$-scalars in polynomial time. However, it seems possible to achieve the same result with purely diagrammatical technics. In fact, applying the rewriting strategy described in Section \ref{sec:3} to a scalar reduces it to a normal form from which we can directly read the interpretation. It seems very probable that this requires only a polynomial number of rewrites.

This provides a way to count perfect matchings without referring to Pfaffian computation, and conversely, it gives a new algorithm to compute Pfaffians based on rewriting.

The FKT algorithm only applies to a specific class of graphs, called Pfaffian graphs, \textbf{i.e.}, the graphs admitting a Pfaffian orientation. In particular, all planar graphs are Pafaffian \cite{kasteleyn1967graph}. It seems that Pfaffian orientiation are directly connected to the behavior of fermionic swap and their lack of naturality which introduces $-1$ weights in the edges. More generally, all graphs not containing $K_{3,3}$ are Pfaffian \cite{DBLP:conf/swat/Vazirani88,little1974extension} (we recall that planar graphs are precisely the graphs not containing neither $K_{3,3}$ nor $K_5 $ as minors). Moreover, there also exists a polynomial time algorithm for $K_5$-minor-free graphs \cite{DBLP:conf/coco/StraubTW14} based on graph decomposition. There is a large amount of work in perspective, re-expressing in diagrammatic terms those different variations and understanding adequately how our rewriting rules could encode the minor constraints.

Formalising and implementing those different algorithms is the object of ongoing work. The main difficulty is to identify the suitable data structures to manipulate the topological data of a given diagram, equivalently, the specific planar embedding of the corresponding graph.

\bibliography{black}

\newpage

\appendix

\begin{proof}[Proof of Lemmas \ref{lem:floop} and \ref{lem:fswap-invol}]
\def\fig{floop-prf}
\begin{align*}
{\setlength{\fboxsep}{0pt}\colorbox{gray!15}{~\strut\input{./figures/\fig/\fig_00.tikz}}}
\eq{}{\setlength{\fboxsep}{0pt}\colorbox{gray!15}{~\strut\input{./figures/\fig/\fig_01.tikz}}}
\eq{}{\setlength{\fboxsep}{0pt}\colorbox{gray!15}{~\strut\input{./figures/\fig/\fig_02.tikz}}}
\eq{}{\setlength{\fboxsep}{0pt}\colorbox{gray!15}{~\strut\input{./figures/\fig/\fig_03.tikz}}}
\eq{}{\setlength{\fboxsep}{0pt}\colorbox{gray!15}{~\strut\input{./figures/\fig/\fig_04.tikz}}}
\eq{}{\setlength{\fboxsep}{0pt}\colorbox{gray!15}{~\strut\input{./figures/\fig/\fig_05.tikz}}}
\end{align*}
\def\fig{fswap-invol-prf}
\begin{align*}
{\setlength{\fboxsep}{0pt}\colorbox{gray!15}{~\strut\input{./figures/\fig/\fig_00.tikz}}}
\eq{}{\setlength{\fboxsep}{0pt}\colorbox{gray!15}{~\strut\input{./figures/\fig/\fig_01.tikz}}}
\eq{}{\setlength{\fboxsep}{0pt}\colorbox{gray!15}{~\strut\input{./figures/\fig/\fig_02.tikz}}}
\eq{}{\setlength{\fboxsep}{0pt}\colorbox{gray!15}{~\strut\input{./figures/\fig/\fig_03.tikz}}}
\eq{}{\setlength{\fboxsep}{0pt}\colorbox{gray!15}{~\strut\input{./figures/\fig/\fig_04.tikz}}}
\eq{}{\setlength{\fboxsep}{0pt}\colorbox{gray!15}{~\strut\input{./figures/\fig/\fig_05.tikz}}}
\end{align*}
\end{proof}

\begin{proof}[Proof of \Cref{prop:rewrites-derivable}]
\def\fig{0-edge-prf}
\begin{align*}
{\setlength{\fboxsep}{0pt}\colorbox{gray!15}{~\strut\input{./figures/\fig/\fig_00.tikz}}}
\eq{}{\setlength{\fboxsep}{0pt}\colorbox{gray!15}{~\strut\input{./figures/\fig/\fig_01.tikz}}}
\eq{}{\setlength{\fboxsep}{0pt}\colorbox{gray!15}{~\strut\input{./figures/\fig/\fig_02.tikz}}}
\eq{}{\setlength{\fboxsep}{0pt}\colorbox{gray!15}{~\strut\input{./figures/\fig/\fig_03.tikz}}}
\eq{}{\setlength{\fboxsep}{0pt}\colorbox{gray!15}{~\strut\input{./figures/\fig/\fig_04.tikz}}}
\eq{}{\setlength{\fboxsep}{0pt}\colorbox{gray!15}{~\strut\input{./figures/\fig/\fig_05.tikz}}}
\eq{}{\setlength{\fboxsep}{0pt}\colorbox{gray!15}{~\strut\input{./figures/\fig/\fig_06.tikz}}}
\eq{}{\setlength{\fboxsep}{0pt}\colorbox{gray!15}{~\strut\input{./figures/\fig/\fig_07.tikz}}}
\end{align*}

\def\fig{rewrite-pivot}
\begin{align*}
{\setlength{\fboxsep}{0pt}\colorbox{gray!15}{~\strut\input{./figures/\fig/\fig_00.tikz}}}
\eq{}r\cdot{\setlength{\fboxsep}{0pt}\colorbox{gray!15}{~\strut\input{./figures/\fig/\fig_01.tikz}}}
\eq{}r\cdot{\setlength{\fboxsep}{0pt}\colorbox{gray!15}{~\strut\input{./figures/\fig/\fig_02.tikz}}}
\eq{}r\cdot{\setlength{\fboxsep}{0pt}\colorbox{gray!15}{~\strut\input{./figures/\fig/\fig_03.tikz}}}
\end{align*}

\def\fig{rewrite-fusion}
\begin{align*}
{\setlength{\fboxsep}{0pt}\colorbox{gray!15}{~\strut\input{./figures/\fig/\fig_00.tikz}}}
\eq{}rs\cdot{\setlength{\fboxsep}{0pt}\colorbox{gray!15}{~\strut\input{./figures/\fig/\fig_01.tikz}}}
\eq{}rs\cdot{\setlength{\fboxsep}{0pt}\colorbox{gray!15}{~\strut\input{./figures/\fig/\fig_02.tikz}}}
\eq{}{\setlength{\fboxsep}{0pt}\colorbox{gray!15}{~\strut\input{./figures/\fig/\fig_03.tikz}}}
\end{align*}

\def\fig{rewrite-fusion-0}
\begin{align*}
{\setlength{\fboxsep}{0pt}\colorbox{gray!15}{~\strut\input{./figures/\fig/\fig_00.tikz}}}
\eq{}{\setlength{\fboxsep}{0pt}\colorbox{gray!15}{~\strut\input{./figures/\fig/\fig_01.tikz}}}
\eq{}{\setlength{\fboxsep}{0pt}\colorbox{gray!15}{~\strut\input{./figures/\fig/\fig_02.tikz}}}
\eq{}{\setlength{\fboxsep}{0pt}\colorbox{gray!15}{~\strut\input{./figures/\fig/\fig_03.tikz}}}
\end{align*}

\def\fig{rewrite-W-loop}
\begin{align*}
{\setlength{\fboxsep}{0pt}\colorbox{gray!15}{~\strut\input{./figures/\fig/\fig_00.tikz}}}
\eq{}{\setlength{\fboxsep}{0pt}\colorbox{gray!15}{~\strut\input{./figures/\fig/\fig_01.tikz}}}
\eq{}r^{\frac12}\cdot{\setlength{\fboxsep}{0pt}\colorbox{gray!15}{~\strut\input{./figures/\fig/\fig_02.tikz}}}
\eq{}r^{\frac12}\cdot{\setlength{\fboxsep}{0pt}\colorbox{gray!15}{~\strut\input{./figures/\fig/\fig_03.tikz}}}
\eq{}r^{\frac12}\cdot{\setlength{\fboxsep}{0pt}\colorbox{gray!15}{~\strut\input{./figures/\fig/\fig_04.tikz}}}
\eq{}{\setlength{\fboxsep}{0pt}\colorbox{gray!15}{~\strut\input{./figures/\fig/\fig_05.tikz}}}
\end{align*}

\def\fig{rewrite-binary-W}
\begin{align*}
{\setlength{\fboxsep}{0pt}\colorbox{gray!15}{~\strut\input{./figures/\fig/\fig_00.tikz}}}
\eq{}r\cdot{\setlength{\fboxsep}{0pt}\colorbox{gray!15}{~\strut\input{./figures/\fig/\fig_01.tikz}}}
\eq{}r\cdot{\setlength{\fboxsep}{0pt}\colorbox{gray!15}{~\strut\input{./figures/\fig/\fig_02.tikz}}}
\end{align*}
\end{proof}

\begin{proof}[Proof of \cref{prop:termination}]
We take each big step of the rewrite strategy, and show that each one terminates into the appropriate form.
\begin{enumerate}
\item We are going to define for every diagram a quantity, as a tuple, whose lexicographic order will be strictly reduced by any of the rewrite step. The very first quantity needed requires some focus.

We say that, whenever the following situation occurs (with $r_i\neq0$ and $s_i\neq0$), the two extremal black nodes are \emph{fusion-equivalent} if they have degree $\geq3$:
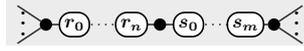
\[\tikzfig{equiv-class}\]
We define an equivalence relation $\sim_f$ between black nodes by taking the reflexive, transitive closure of this relation. We say that an equivalence class is internal if \emph{all} its constituents are internal nodes.

For a diagram $D$, we define $T(D)\in\mathbb N^6$ as the quantity:
\[ \small T(D):=\left(\#\left(
\begin{array}{c}
\text{internal}\\
\sim_f\text{-classes}
\end{array}
\right), \#\left(\tikzfig{W-spider},n\geq 3\right), \#\left(\tikzfig{edge}\right), \#\left(\tikzfig{X}\right), \#\left(\tikzfig{fswap}\right), \#\left(\def\fig{rewrite-zero}{\setlength{\fboxsep}{0pt}\colorbox{gray!15}{~\strut\input{./figures/\fig/\fig_00.tikz}}}\right)\right)\]
We can show that if $D\to D'$ then $T(D')<T(D)$ in the lexicographic order:
\begin{itemize}
\item $\def\fig{ax-sum}{\setlength{\fboxsep}{0pt}\colorbox{gray!15}{~\strut\input{./figures/\fig/\fig_00.tikz}}}\to{\setlength{\fboxsep}{0pt}\colorbox{gray!15}{~\strut\input{./figures/\fig/\fig_01.tikz}}}$ may reduce may reduce the first two values (if one of the black nodes initially has degree 3), but it at least reduces the number of edges.
\item $\def\fig{rewrite-W-loop}{\setlength{\fboxsep}{0pt}\colorbox{gray!15}{~\strut\input{./figures/\fig/\fig_00.tikz}}}\to{\setlength{\fboxsep}{0pt}\colorbox{gray!15}{~\strut\input{./figures/\fig/\fig_05.tikz}}}$ similarly may reduce the first two values, but at least the third one (the number of edges).
\item $\def\fig{0-edge-prf}{\setlength{\fboxsep}{0pt}\colorbox{gray!15}{~\strut\input{./figures/\fig/\fig_00.tikz}}}\to{\setlength{\fboxsep}{0pt}\colorbox{gray!15}{~\strut\input{./figures/\fig/\fig_07.tikz}}}$ can only reduce the third value (and does not change the others).
\item $\def\fig{rewrite-fswap-removal}{\setlength{\fboxsep}{0pt}\colorbox{gray!15}{~\strut\input{./figures/\fig/\fig_00.tikz}}}\to{\setlength{\fboxsep}{0pt}\colorbox{gray!15}{~\strut\input{./figures/\fig/\fig_01.tikz}}}$ and 
$\def\fig{rewrite-floop}{\setlength{\fboxsep}{0pt}\colorbox{gray!15}{~\strut\input{./figures/\fig/\fig_00.tikz}}}\to{\setlength{\fboxsep}{0pt}\colorbox{gray!15}{~\strut\input{./figures/\fig/\fig_01.tikz}}}$ both reduce the fifth value without changing the other ones
\item $\def\fig{ax-phase-fusion}{\setlength{\fboxsep}{0pt}\colorbox{gray!15}{~\strut\input{./figures/\fig/\fig_00.tikz}}}\to{\setlength{\fboxsep}{0pt}\colorbox{gray!15}{~\strut\input{./figures/\fig/\fig_01.tikz}}}$ can only reduce the third value.
\item $\def\fig{rewrite-binary-W}{\setlength{\fboxsep}{0pt}\colorbox{gray!15}{~\strut\input{./figures/\fig/\fig_00.tikz}}}\underset{r\neq0}\to r\cdot{\setlength{\fboxsep}{0pt}\colorbox{gray!15}{~\strut\input{./figures/\fig/\fig_02.tikz}}}$ only changes the number of binary black nodes (in particular, it cannot create new $\sim_f$-classes).
\item $\def\fig{rewrite-fusion-0}{\setlength{\fboxsep}{0pt}\colorbox{gray!15}{~\strut\input{./figures/\fig/\fig_00.tikz}}}\to r\cdot{\setlength{\fboxsep}{0pt}\colorbox{gray!15}{~\strut\input{./figures/\fig/\fig_03.tikz}}}$ may decrease the two first values, but necessarily decreases the third one.
\item $\def\fig{rewrite-fusion}{\setlength{\fboxsep}{0pt}\colorbox{gray!15}{~\strut\input{./figures/\fig/\fig_00.tikz}}}\overset{\substack{n\geq2\\m\geq2}}{\underset{\substack{r\neq0\\s\neq0}}\to}{\setlength{\fboxsep}{0pt}\colorbox{gray!15}{~\strut\input{./figures/\fig/\fig_03.tikz}}}$ : since $n\geq2$ and $m\geq 2$, the two extremal black nodes are necessarily in the same $\sim_f$-class. In that case, the number of such classes does not change, but the second value is decreased.
Notice that, when applied repeatedly, this rule eventually reduces every class to a single element.
\item $\def\fig{rewrite-zero}{\setlength{\fboxsep}{0pt}\colorbox{gray!15}{~\strut\input{./figures/\fig/\fig_00.tikz}}}\to{\setlength{\fboxsep}{0pt}\colorbox{gray!15}{~\strut\input{./figures/\fig/\fig_01.tikz}}}$ reduces the last value.
\item $\def\fig{rewrite-pivot}{\setlength{\fboxsep}{0pt}\colorbox{gray!15}{~\strut\input{./figures/\fig/\fig_00.tikz}}}\underset{\substack{n\neq1\\m\neq1}}\to r\cdot{\setlength{\fboxsep}{0pt}\colorbox{gray!15}{~\strut\input{./figures/\fig/\fig_03.tikz}}}$ requires a case distinction. If $n=m=0$, we reduce the number of edges. If $n=0$ and $m\geq2$ if the bottom  black node is internal, we decrease the first value, if not, the second value (similarly, the case $n\geq2$ and $m=0$ decreases $T$). In the case where $n\geq2$ and $m\geq2$, let us focus on one of the two black nodes. Since the rule can only be applied when none other can, all its neighbours each constitute their own $\sim_f$-class. After the rewrite, all the top non-binary black nodes will join the $\sim_f$-class of the neighbour they connect to, so the number of \emph{internal} $\sim_f$-classes is either unchanged if the node was not internal, or reduced by one if it was. Since at least one of the two nodes had to be internal for the rewrite to apply, the overall number of $\sim_f$-classes reduces.
\end{itemize}
At the end of these rewrites, the diagram does not have any internal node left (all were removed by Rule $(\ast)$ or fused into boundary nodes). The only possible form of a diagram with no internal nodes is the pseudo-WGS-X form.
Notice that in between applications of $(\ast)$, which decreases the number of internal $\sim_f$-classes, there can only be a linear (in the size of the diagram) number of rewrites that can be applied. When $(\ast)$ is applied, there is no more than $O(n^2)$ generators ($n$ being the number of black nodes). This forces the step to stop after a polynomial amount of time.
\item Consider a boundary node that is not in the proper form for the diagram to be in WGS-X form. In all generality, the node's neighbourhood is in the form (up to rearranging of the output wires):
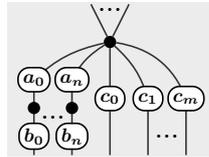
\def\fig{pseudo-to-WGS-X}
\begin{align*}
{\setlength{\fboxsep}{0pt}\colorbox{gray!15}{~\strut\input{./figures/\fig/\fig_00.tikz}}}
\end{align*}
We then apply two black nodes on the $m$ rightmost outputs:
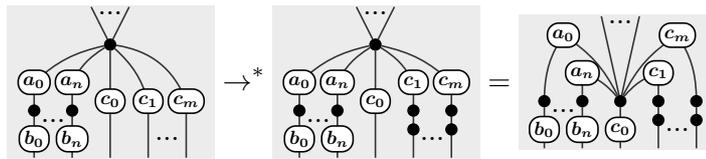
\begin{align*}
{\setlength{\fboxsep}{0pt}\colorbox{gray!15}{~\strut\input{./figures/\fig/\fig_00.tikz}}}
\to^*{\setlength{\fboxsep}{0pt}\colorbox{gray!15}{~\strut\input{./figures/\fig/\fig_01.tikz}}}
\eq{}{\setlength{\fboxsep}{0pt}\colorbox{gray!15}{~\strut\input{./figures/\fig/\fig_02.tikz}}}
\end{align*}
Doing this for all ``improper'' nodes allows us to associate exactly one type-1 boundary node to each output. At this point, we can move all $\tikzfig{fswap}$ up into the graph part of the diagram. This changes the order of the vertices, and potentially adds $-1$ weights on some edges, potentially on output wires. The weights on output wires are precisely handled by the last substep:
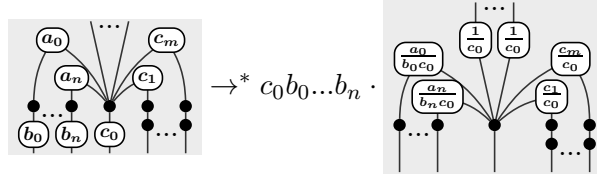
\begin{align*}
{\setlength{\fboxsep}{0pt}\colorbox{gray!15}{~\strut\input{./figures/\fig/\fig_02.tikz}}}
\to^*c_0b_0...b_n\cdot{\setlength{\fboxsep}{0pt}\colorbox{gray!15}{~\strut\input{./figures/\fig/\fig_03.tikz}}}
\end{align*}
When all weights are moved up, the diagram ends up in WGS-X form. This step obviously stops after a polynomial amount of time.
\item The rule $(\ast)$ together with that of \Cref{fig:rewrite-reduce} cover all situations when a type-0 boundary node has a right neighbour (notice that boundary nodes of arity $1$ cannot have neighbours). In the process of this step, we are rewriting the WGS-X state in an equivalent one, with fewer type-0 nodes with a right neighbour. The step hence eventually terminates, and in such a situation that none of the type-0 boundary nodes have a right neighbour. This is exactly the form of a reduced WGS-X state. For the same reason as in step 1, this step terminates in polynomial time.
\end{enumerate}
Finally, you may notice that the rewrite strategy terminates in polynomial time, as each of the three steps is polynomial.
\end{proof}

\begin{proof}[Proof of \Cref{lem:vertex-removal}]
W.l.o.g.~we can assume $i=1$. We then get:
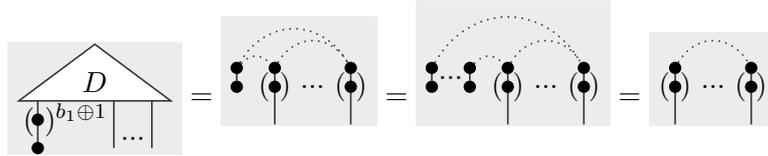
\def\fig{WGS-X-vertex-removal}
\begin{align*}
{\setlength{\fboxsep}{0pt}\colorbox{gray!15}{~\strut\input{./figures/\fig/\fig_00.tikz}}}
\eq{}{\setlength{\fboxsep}{0pt}\colorbox{gray!15}{~\strut\input{./figures/\fig/\fig_01.tikz}}}
\eq{}{\setlength{\fboxsep}{0pt}\colorbox{gray!15}{~\strut\input{./figures/\fig/\fig_02.tikz}}}
\eq{}{\setlength{\fboxsep}{0pt}\colorbox{gray!15}{~\strut\input{./figures/\fig/\fig_03.tikz}}}
\end{align*}
where dotted lines represent potential edges (potentially with weights), that cross with fermionic swaps $\tikzfig{fswap}$; 
and where we used the following steps:
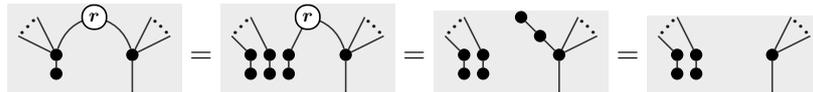
\def\fig{completeness-details}
\begin{align*}
{\setlength{\fboxsep}{0pt}\colorbox{gray!15}{~\strut\input{./figures/\fig/\fig_00.tikz}}}
\eq{}{\setlength{\fboxsep}{0pt}\colorbox{gray!15}{~\strut\input{./figures/\fig/\fig_01.tikz}}}
\eq{}{\setlength{\fboxsep}{0pt}\colorbox{gray!15}{~\strut\input{./figures/\fig/\fig_02.tikz}}}
\eq{}{\setlength{\fboxsep}{0pt}\colorbox{gray!15}{~\strut\input{./figures/\fig/\fig_03.tikz}}}
\end{align*}
together with the fact that fermionic swaps do not interfere here as:
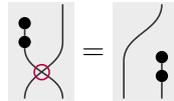
\def\fig{ket0-through-fswap}
\begin{align*}
{\setlength{\fboxsep}{0pt}\colorbox{gray!15}{~\strut\input{./figures/\fig/\fig_00.tikz}}}
\eq{}{\setlength{\fboxsep}{0pt}\colorbox{gray!15}{~\strut\input{./figures/\fig/\fig_01.tikz}}}
\end{align*}
\end{proof}

\begin{proof}[Proof of \Cref{lem:WGS-X-zero}]
The right to left implication is obvious. Now suppose $s\neq0$. Then, applying $\tikzfig{bra1}^{\otimes n}\circ \tikzfig{X}^{\otimes(\vec b_1\oplus1...1)}$ to the outputs reduces to the WGS-X state $(s, G_\emptyset, () )$ thanks to \Cref{lem:vertex-removal} (where $G_\emptyset = ([],\emptyset,\_)$ is the empty graph), whose interpretation is scalar $s$. Since $s\neq0$, $\interp{D}\neq0$.
\end{proof}

\begin{proof}[Proof of \Cref{lem:WGS-X-zero-equivalence}]
We can start by removing all edges from the WGS-X state thanks to:
\def\fig{zero-deriv}
\begin{align*}
0\cdot{\setlength{\fboxsep}{0pt}\colorbox{gray!15}{~\strut\input{./figures/\fig/\fig_00.tikz}}}
\eq{}0\cdot{\setlength{\fboxsep}{0pt}\colorbox{gray!15}{~\strut\input{./figures/\fig/\fig_01.tikz}}}
\eq{}{\setlength{\fboxsep}{0pt}\colorbox{gray!15}{~\strut\input{./figures/\fig/\fig_02.tikz}}}
\eq{}{\setlength{\fboxsep}{0pt}\colorbox{gray!15}{~\strut\input{./figures/\fig/\fig_03.tikz}}}
\eq{}0\cdot{\setlength{\fboxsep}{0pt}\colorbox{gray!15}{~\strut\input{./figures/\fig/\fig_04.tikz}}}
\eq{}0\cdot{\setlength{\fboxsep}{0pt}\colorbox{gray!15}{~\strut\input{./figures/\fig/\fig_05.tikz}}}
\end{align*}
and then fusing black nodes together when possible (again, the fermionic swaps do not cause issues here). When done, we end up with a tensor of $\tikzfig{ket0}$s and $\tikzfig{ket1}$s. It remains to show the following:
\def\fig{zero-ket}
\begin{align*}
0\cdot{\setlength{\fboxsep}{0pt}\colorbox{gray!15}{~\strut\input{./figures/\fig/\fig_00.tikz}}}
\eq{}{\setlength{\fboxsep}{0pt}\colorbox{gray!15}{~\strut\input{./figures/\fig/\fig_01.tikz}}}
\eq{}{\setlength{\fboxsep}{0pt}\colorbox{gray!15}{~\strut\input{./figures/\fig/\fig_02.tikz}}}
\eq{}0\cdot{\setlength{\fboxsep}{0pt}\colorbox{gray!15}{~\strut\input{./figures/\fig/\fig_03.tikz}}}
\eq{}0\cdot{\setlength{\fboxsep}{0pt}\colorbox{gray!15}{~\strut\input{./figures/\fig/\fig_04.tikz}}}
\eq{}0\times 0\cdot{\setlength{\fboxsep}{0pt}\colorbox{gray!15}{~\strut\input{./figures/\fig/\fig_05.tikz}}}
\eq{}0\cdot{\setlength{\fboxsep}{0pt}\colorbox{gray!15}{~\strut\input{./figures/\fig/\fig_06.tikz}}}
\end{align*}
\end{proof}

\begin{proof}[Proof of \Cref{thm:completeness}]
The right to left implication is soundness of the rules, which is obvious as the equational theory is a subpart of the equational theory for ZW-diagrams, which is known to be sound \cite{DBLP:conf/lics/Hadzihasanovic15}.
Let us now assume that $\interp{D_1} = \interp{D_2}$. Using \Cref{lem:rewrite-to-WGS-X} and \Cref{prop:rewrites-derivable}, we can turn both $D_1$ and $D_2$ into respectively $d_1$ and $d_2$, in rWGS-X form, and in a way that preserves semantics, i.e.~$\interp{d_1}=\interp{d_2}$. Let us denote $(s_i, G_i=([1,n], E_i, w_i), \vec b_i)$ the rWGS-X form of $d_i$.

Notice that if $\interp{d_1}=\interp{d_2}=0$, \Cref{lem:WGS-X-zero-equivalence} ensures that $d_1$ and $d_2$ can both be put in the same form, which proves the result in that case. In the following, we can hence consider that the diagrams have non-$0$ interpretation.

First, let us show that $\vec b_1=\vec b_2$. Let us consider the first index $i$ for which ${b_1}_i\neq {b_2}_i$. As both $d_1$ and $d_2$ are in reduced form, for all $j<i$, ${b_1}_j = {b_2}_j = 1$. Suppose w.l.o.g.~that ${b_1}_i = 0$ and ${b_2}_i = 1$. Again, since they are in reduced form, the $i$th vertex in $G_1$ can only have neigbours on its left. Let us apply $\tikzfig{bra1}$ to the first $i$ qubits in $d_1$ and $d_2$. On the one hand:
\def\fig{completeness-d1}
\begin{align*}
{\setlength{\fboxsep}{0pt}\colorbox{gray!15}{~\strut\input{./figures/\fig/\fig_00.tikz}}}
&\eq{}s_1\cdot{\setlength{\fboxsep}{0pt}\colorbox{gray!15}{~\strut\input{./figures/\fig/\fig_01.tikz}}}
\eq{}s_1\cdot{\setlength{\fboxsep}{0pt}\colorbox{gray!15}{~\strut\input{./figures/\fig/\fig_03.tikz}}}
\eq{} 0
\end{align*}
where we used \Cref{lem:vertex-removal} for the first $i-1$ qubits and the black node fusion for the $i$th.
The map hence obtained from $d_1$ is the zero map. On the other hand, using \Cref{lem:vertex-removal} again:
\def\fig{completeness-d2}
\begin{align*}
{\setlength{\fboxsep}{0pt}\colorbox{gray!15}{~\strut\input{./figures/\fig/\fig_00.tikz}}}
&\eq{}s_2\cdot{\setlength{\fboxsep}{0pt}\colorbox{gray!15}{~\strut\input{./figures/\fig/\fig_01.tikz}}}
\eq{}s_2\cdot{\setlength{\fboxsep}{0pt}\colorbox{gray!15}{~\strut\input{./figures/\fig/\fig_04.tikz}}}
\end{align*}
which thanks to \Cref{lem:WGS-X-zero} is necessarily not zero. There is a contradiction, hence, if $d_1$ and $d_2$ are in rWGS-X form with $\interp{d_1} = \interp{d_2}$, then $\vec b_1 = \vec b_2$.

We can then show that we can recover $s_1=s_2$. Indeed, if we apply $\tikzfig{bra1}^{\otimes n}\circ \tikzfig{X}^{\vec b_1\oplus1...1}$ on both diagrams, we get (\Cref{lem:vertex-removal}):
\def\fig{completeness-getting-scalar}
\begin{align*}
{\setlength{\fboxsep}{0pt}\colorbox{gray!15}{~\strut\input{./figures/\fig/\fig_00.tikz}}}
\eq{}s_1\cdot{\setlength{\fboxsep}{0pt}\colorbox{gray!15}{~\strut\input{./figures/\fig/\fig_01.tikz}}}
\eq{}s_1\cdot{\setlength{\fboxsep}{0pt}\colorbox{gray!15}{~\strut\input{./figures/\fig/\fig_02.tikz}}}
\end{align*}
Similarly on $d_2$, we get $s_2\cdot{\setlength{\fboxsep}{0pt}\colorbox{gray!15}{~\strut\input{./figures/\fig/\fig_02.tikz}}}$, which proves $s_1=s_2$.

Finally, we can show that the weight between every pair $(i,j)$ of vertices in $G_1$ and $G_2$ is the same, with the convention that having no edge between two vertices is equivalent to having an edge with weight $0$:
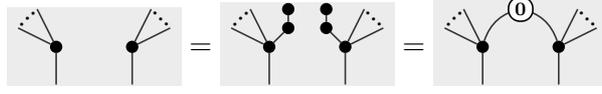
\def\fig{0-edge-is-no-edge}
\begin{align*}
{\setlength{\fboxsep}{0pt}\colorbox{gray!15}{~\strut\input{./figures/\fig/\fig_00.tikz}}}
\eq{}{\setlength{\fboxsep}{0pt}\colorbox{gray!15}{~\strut\input{./figures/\fig/\fig_01.tikz}}}
\eq{}{\setlength{\fboxsep}{0pt}\colorbox{gray!15}{~\strut\input{./figures/\fig/\fig_02.tikz}}}
\end{align*}
To that end, consider $\vec b' := \vec b_1\oplus1..1\overset{{\color{gray}i}}01...1\overset{{\color{gray}j}}01...1$. We may now apply $\tikzfig{bra1}^{\otimes n}\circ \tikzfig{X}^{\vec b'}$ on both diagrams, to get on the one hand:
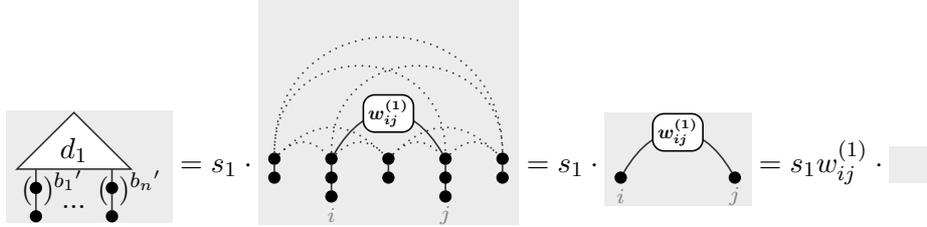
\def\fig{completeness-edge-ij}
\begin{align*}
{\setlength{\fboxsep}{0pt}\colorbox{gray!15}{~\strut\input{./figures/\fig/\fig_00.tikz}}}
\eq{}s_1\cdot{\setlength{\fboxsep}{0pt}\colorbox{gray!15}{~\strut\input{./figures/\fig/\fig_01.tikz}}}
\eq{}s_1\cdot{\setlength{\fboxsep}{0pt}\colorbox{gray!15}{~\strut\input{./figures/\fig/\fig_02.tikz}}}
\eq{}s_1 w_{ij}^{(1)}\cdot{\setlength{\fboxsep}{0pt}\colorbox{gray!15}{~\strut\input{./figures/\fig/\fig_03.tikz}}}
\end{align*}
and similary on the other hand, we get $s_1 w_{ij}^{(2)}$ from $d_2$. This implies $w_{ij}^{(1)}=w_{ij}^{(2)}$. Doing so for all pairs of vertices in $G_i$ gives us $G_1=G_2$.
\end{proof}

\begin{proof}[Proof of \Cref{lem:mtensor}]
	Given $\alpha \in \{0,1\}^{a+c+d+b} $ we write $\alpha= \alpha^1 \cdot \alpha^2 \cdot \alpha^3 $ with $\alpha^1 \in \{0,1\}^{a}$, $\alpha^2 \in \{0,1\}^{c+d}$ and $\alpha^3  \in \{0,1\}^{b} $. We have $[f\otimes g]_\alpha= [f]_{\alpha^1 \cdot \alpha^3 } [g]_{\alpha^2}$. Pictorially:
	\def\fig{tenssateform}
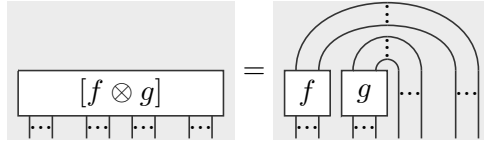
	\begin{align*}
		{\setlength{\fboxsep}{0pt}\colorbox{gray!15}{~\strut\input{./figures/\fig/\fig_00.tikz}}}
		\eq{}{\setlength{\fboxsep}{0pt}\colorbox{gray!15}{~\strut\input{./figures/\fig/\fig_01.tikz}}}
	\end{align*}
	
	We compute the matchgate identity corresponding to $\alpha, \beta \in \{0,1\}^{a+c+d+b}$.
	
	\[\sum_{k=1 }^{|\alpha \oplus \beta |} (-1)^k [f\otimes g]_{\alpha \oplus e_{p_k}} [f\otimes g]_{\beta \oplus e_{p_k} }\]
	
	We start by splitting the sum into three terms, one for $a$, one for $c+d$ and one for $b$. By doing so, we also re-index the $ p_k $ for each sum:
	
	\begin{align*}
		&[g]_{\alpha^2}[g]_{\beta^2}\left(\sum_{k=1 }^{ |\alpha^1 \oplus \beta^1 | } (-1)^k [f]_{(\alpha^1 \oplus e_{p_k} ) \cdot \alpha^3 } [f]_{(\beta^1 \oplus e_{p_k} ) \cdot \beta^3 }\right)\\
		+& (-1)^{|\alpha^1 \oplus \beta^1 |}[f]_{\alpha^1 \cdot \alpha^3 } [f]_{\beta^1 \cdot \beta^3 }\left(\sum_{k=1 }^{ |\alpha^2 \oplus \beta^2 | } (-1)^k [g]_{\alpha^2 \oplus _{p_k} }[g]_{\beta^2 \oplus _{p_k} }\right)\\
		+& (-1)^{|\alpha^1 \oplus \beta^1 | + |\alpha^2 \oplus \beta^2 | }[g]_{\alpha^2}[g]_{\beta^2}\left(\sum_{k=1}^{|\alpha^3 \oplus \beta^3|} (-1)^k [f]_{\alpha^1  \cdot (\alpha^3 \oplus e_{p_k} )} [f]_{\beta^1 \cdot (\beta^3 \oplus e_{p_k} ) }\right)
	\end{align*}
	
	The second term is zero as we recognise the matchgate identity satisfied by $[g]$ for $\alpha^2 $ and $\beta^2 $. Furthermore if $|\alpha^2 \oplus \beta^2 |_2 \neq 0 $ the parity condition on $[g]$ implies that $[g]_{\alpha^2}[g]_{\beta^2}=0$ canceling the first and third terms. In the case where $|\alpha^2 \oplus \beta^2 |_2 = 0 $, we can gather the first and third terms to get: 
	
	\[ [g]_{\alpha^2}[g]_{\beta^2}\left(\sum_{k=1}^{|(\alpha^1 \cdot \alpha^3) \oplus (\beta^1 \cdot \beta^3 ) |} (-1)^k [f]_{(\alpha^1 \cdot \alpha^3) \oplus e_{p_k } } [f]_{(\beta^1 \cdot \beta^3) \oplus e_{p_k } }\right) \]
	
	We now recognise the matchgate identity satisfied by $[f]$ for $\alpha^1 \cdot \alpha^3$ and $\beta^1 \cdot \beta^3$, which evaluates to zero. So $[f\otimes g] $ satisfies the matchgate identities.
	
\end{proof}

\begin{proof}[Proof of \Cref{lem:mcomp}]
	
	Let $i,i+1 \in \{1,...,n+2\}$ be two consecutive indices. We write $\alpha[xy] \in \{0,1\}^{n+2}$ for the binary word $\alpha \in \{0,1\}^n $ in which $x,y \in \{0,1\}$ have been respectively inserted in position $i$ and $i+1$. Denoting $\Gamma'$ the tensor obtained by contracting the indices $i$ and $i+1$ in $\Gamma$ we have: $\Gamma'_{\alpha} = \Gamma_{\alpha[00]} + \Gamma_{\alpha[11]}$.
	
	We now compute a matchgate identity: 
	
	\begin{align*}& \sum_{k=1}^{|\alpha \oplus \beta|} (-1)^k \Gamma'_{\alpha \oplus e_{p_k }} \Gamma'_{\beta \oplus e_{p_k } }= \\
          &\sum_{k=1}^{|\alpha \oplus \beta|}  (-1)^k \left(\Gamma_{(\alpha\oplus e_{p_k })[00] } + \Gamma_{(\alpha\oplus e_{p_k })[11] } \right) \left(\Gamma_{(\beta\oplus e_{p_k })[00] } + \Gamma_{(\beta\oplus e_{p_k })[11] } \right) \end{align*}
	
	Distributing and using $(\alpha \oplus \beta) [xy]= \alpha[xy] \oplus \beta[00] $ gives us four terms:
	
	\begin{align*}
		&\sum_{k=1}^{|\alpha \oplus \beta|}  (-1)^k \Gamma_{\alpha[00] \oplus e_{p_k }[00]}\Gamma_{\beta[00]\oplus e_{p_k } [00] }\\
		+& \sum_{k=1}^{|\alpha \oplus \beta|} (-1)^k \Gamma_{\alpha[00] \oplus e_{p_k }[00]}\Gamma_{\beta[11]\oplus e_{p_k } [00] } \\
		+& \sum_{k=1}^{|\alpha \oplus \beta|} (-1)^k \Gamma_{\alpha[11] \oplus e_{p_k }[00]}\Gamma_{\beta[00]\oplus e_{p_k } [00] } \\
		+&  \sum_{k=1}^{|\alpha \oplus \beta|} (-1)^k \Gamma_{\alpha[11] \oplus e_{p_k }[00]}\Gamma_{\beta[11]\oplus e_{p_k } [00] }
	\end{align*}
	
	The first and last terms correspond to matchgate identities respectively for $\alpha[00]$ and $\beta[00]$, and for $\alpha[11]$ and $\beta[11]$, so they are zero as $\Gamma$ is required to satisfy the matchgate identities. We would like to say the same about the second and third terms. Sadly some terms are missing to get complete matchgate identities since we have added new positions where the words differ. The missing terms are $(-1)^{\ell} \Gamma_{\alpha[10] }\Gamma_{\beta[01]} $ and $(-1)^{\ell +1} \Gamma_{\alpha[01] }\Gamma_{\beta[10]} $ for the second line, and $(-1)^{\ell} \Gamma_{\alpha[01] }\Gamma_{\beta[10]} $ and $(-1)^{\ell +1} \Gamma_{\alpha[10] }\Gamma_{\beta[01]} $ for the third line, where $\ell$ is the index of the position $i$ in the set of differing positions. We see that the four missing terms altogether cancel each other. So we can safely add the missing terms in the sum and get two complete matchgate identities, respectively, for $\alpha[00]$ and $\beta[11]$, and for $\alpha[11]$ and $\beta[00]$. Finally, the global sum is zero, and $\Gamma'$ satisfies the matchgate identities. 
	
\end{proof}

\end{document}

%% file: figures/tenssateform/tenssateform_01.tikz
\begin{tikzpicture}[scale=0.6]
	\begin{pgfonlayer}{nodelayer}
		\node [style=none] (5) at (-2.125, -1) {};
		\node [style=none] (6) at (-1.125, -1) {};
		\node [style=none] (7) at (-1.125, 0) {};
		\node [style=none] (8) at (-2.125, 0) {};
		\node [style=none] (9) at (-1.625, -0.5) {$f$};
		\node [style=none] (21) at (-1.625, -1.25) {$...$};
		\node [style=none] (22) at (-1.875, -1) {};
		\node [style=none] (23) at (-1.375, -1) {};
		\node [style=none] (24) at (-1.875, -1.5) {};
		\node [style=none] (25) at (-1.375, -1.5) {};
		\node [style=none] (27) at (-1.875, 0) {};
		\node [style=none] (28) at (-1.375, 0) {};
		\node [style=none] (32) at (1.625, 0) {};
		\node [style=none] (33) at (2.125, 0) {};
		\node [style=none] (34) at (1.625, -1.5) {};
		\node [style=none] (35) at (2.125, -1.5) {};
		\node [style=none] (36) at (1.875, -0.5) {$...$};
		\node [style=none] (37) at (0.125, 1) {};
		\node [style=none] (38) at (0.125, 1.5) {};
		\node [style=none] (39) at (0.125, 1.25) {$\bvdots$};
		\node [style=none] (40) at (-0.875, -1) {};
		\node [style=none] (41) at (0.125, -1) {};
		\node [style=none] (42) at (0.125, 0) {};
		\node [style=none] (43) at (-0.875, 0) {};
		\node [style=none] (44) at (-0.375, -0.5) {$g$};
		\node [style=none] (45) at (-0.375, -1.25) {$...$};
		\node [style=none] (46) at (-0.625, -1) {};
		\node [style=none] (47) at (-0.125, -1) {};
		\node [style=none] (48) at (-0.625, -1.5) {};
		\node [style=none] (49) at (-0.125, -1.5) {};
		\node [style=none] (51) at (-0.625, 0) {};
		\node [style=none] (52) at (-0.125, 0) {};
		\node [style=none] (53) at (-0.625, 0) {};
		\node [style=none] (54) at (-0.125, 0) {};
		\node [style=none] (56) at (0.375, 0) {};
		\node [style=none] (57) at (0.875, 0) {};
		\node [style=none] (58) at (0.375, -1.5) {};
		\node [style=none] (59) at (0.875, -1.5) {};
		\node [style=none] (60) at (0.625, -0.5) {$...$};
		\node [style=none] (61) at (0.125, 0.25) {};
		\node [style=none] (62) at (0.125, 0.75) {};
		\node [style=none] (63) at (0.125, 0.5) {$\bvdots$};
	\end{pgfonlayer}
	\begin{pgfonlayer}{edgelayer}
		\draw [fill=white] (6.center)
			 to (7.center)
			 to (8.center)
			 to (5.center)
			 to cycle;
		\draw (24.center) to (22.center);
		\draw (25.center) to (23.center);
		\draw [in=180, out=90] (27.center) to (38.center);
		\draw [in=360, out=90] (32.center) to (37.center);
		\draw (34.center) to (32.center);
		\draw (35.center) to (33.center);
		\draw [in=-270, out=180] (37.center) to (28.center);
		\draw [in=90, out=0] (38.center) to (33.center);
		\draw [fill=white] (40.center)
			 to (41.center)
			 to (42.center)
			 to (43.center)
			 to cycle;
		\draw (48.center) to (46.center);
		\draw (49.center) to (47.center);
		\draw [bend left=45] (51.center) to (62.center);
		\draw (53.center) to (51.center);
		\draw (54.center) to (52.center);
		\draw (58.center) to (56.center);
		\draw (59.center) to (57.center);
		\draw [bend right=45] (61.center) to (52.center);
		\draw [bend left=45] (61.center) to (56.center);
		\draw [bend left=45] (62.center) to (57.center);
	\end{pgfonlayer}
\end{tikzpicture}

%% file: figures/tenssateform/tenssateform_00.tikz
\begin{tikzpicture}[scale=0.6]
	\begin{pgfonlayer}{nodelayer}
		\node [style=none] (0) at (-2.25, -1) {};
		\node [style=none] (1) at (2.25, -1) {};
		\node [style=none] (2) at (2.25, 0) {};
		\node [style=none] (3) at (-2.25, 0) {};
		\node [style=none] (4) at (0, -0.5) {$[f\otimes g]$};
		\node [style=none] (10) at (-1.75, -1.25) {$...$};
		\node [style=none] (11) at (-0.5, -1.25) {$...$};
		\node [style=none] (13) at (-2, -1) {};
		\node [style=none] (14) at (-1.5, -1) {};
		\node [style=none] (15) at (-2, -1.5) {};
		\node [style=none] (16) at (-1.5, -1.5) {};
		\node [style=none] (17) at (-0.75, -1.5) {};
		\node [style=none] (18) at (-0.75, -1) {};
		\node [style=none] (19) at (-0.25, -1.5) {};
		\node [style=none] (20) at (-0.25, -1) {};
		\node [style=none] (64) at (0.5, -1.25) {$...$};
		\node [style=none] (65) at (0.25, -1.5) {};
		\node [style=none] (66) at (0.25, -1) {};
		\node [style=none] (67) at (0.75, -1.5) {};
		\node [style=none] (68) at (0.75, -1) {};
		\node [style=none] (69) at (1.75, -1.25) {$...$};
		\node [style=none] (70) at (1.5, -1.5) {};
		\node [style=none] (71) at (1.5, -1) {};
		\node [style=none] (72) at (2, -1.5) {};
		\node [style=none] (73) at (2, -1) {};
		\node [style=none] (74) at (-0.5, 1.5) {};
	\end{pgfonlayer}
	\begin{pgfonlayer}{edgelayer}
		\draw [fill=white] (3.center)
			 to (0.center)
			 to (1.center)
			 to (2.center)
			 to cycle;
		\draw (15.center) to (13.center);
		\draw (16.center) to (14.center);
		\draw (17.center) to (18.center);
		\draw (19.center) to (20.center);
		\draw (65.center) to (66.center);
		\draw (67.center) to (68.center);
		\draw (70.center) to (71.center);
		\draw (72.center) to (73.center);
	\end{pgfonlayer}
\end{tikzpicture}

%% file: figures/example-graph-state/example-graph-state_00.tikz
\begin{tikzpicture}[scale=0.7]
	\begin{pgfonlayer}{nodelayer}
		\node [style=dot] (0)  at (-2.25, 0.125) {};
		\node [style=dot] (1)  at (-1.25, -0.625) {};
		\node [style=dot] (2)  at (0.0, -0.875) {};
		\node [style=dot] (3)  at (1.25, -0.625) {};
		\node [style=dot] (4)  at (2.25, 0.125) {};
		\node [style=none, font={\scriptsize}] (5)  at (0.0, 0.875) {$i$};
		\node [style=none, font={\scriptsize}] (6)  at (1.0, 0.125) {$-1$};
		\node [style=none, font={\scriptsize}] (7)  at (-2.0, -0.375) {$2$};
	\end{pgfonlayer}
	\begin{pgfonlayer}{edgelayer}
		\draw (0) to (1);
		\draw [bend left] (0) to (2);
		\draw [bend left] (1) to (3);
		\draw [bend left] (2) to (4);
		\draw (2) to (3);
		\draw [bend right=60] (4) to (1);
		\draw (4) to (3);
	\end{pgfonlayer}
\end{tikzpicture}

%% file: figures/diagfact.tikz
\begin{tikzpicture}
	\begin{pgfonlayer}{nodelayer}
		\node [style=box-no-outline] (153) at (3, -0.75) {$\textbf{Match}$};
		\node [style=box-no-outline] (154) at (0.5, 0.5) {$~\pW~$};
		\node [style=box-no-outline] (155) at (3, 0.5) {$\textbf{Qubit}$};
		\node [style=none] (162) at (1.75, 0.75) {$\interp{\_}$};
	\end{pgfonlayer}
	\begin{pgfonlayer}{edgelayer}
		\draw [style={arrows={->[]}}] (154) to (155);
		\draw [style=dashed arrow] (154) to (153);
		\draw [style=hook arrow] (153) to (155);
	\end{pgfonlayer}
\end{tikzpicture}